\documentclass[a4paper,12pt]{scrartcl}

\usepackage[utf8]{inputenc}
\usepackage[T1]{fontenc}
\usepackage{microtype}
\usepackage{lmodern}
\usepackage{enumitem}

\usepackage[fleqn]{amsmath}
\usepackage{amsthm}
\usepackage{amsfonts}
\usepackage{authblk}

\usepackage{tikz}
\usetikzlibrary{automata, positioning, arrows, shapes, arrows.meta}

\usepackage{thm-restate}
\usepackage{graphicx}
\usepackage{xcolor}
\usepackage{xspace}
\usepackage{amssymb}
\usepackage[most]{tcolorbox}
\usepackage{wrapfig}
\usepackage{multirow}
\usepackage[strings]{underscore}
\usepackage{booktabs}
\usepackage{multicol}
\usepackage{pbox}
\usepackage{enumitem}

\usepackage{algorithm}
\usepackage[noend]{algpseudocode}
\algrenewcommand\algorithmicindent{1.2em}
\algnewcommand\Input{\item[\textbf{Input:}]}
\algnewcommand\Output{\item[\textbf{Output:}]}
\algnewcommand\Effect{\item[\textbf{Effect:}]}
\allowdisplaybreaks
 
\usepackage{xcolor}
\usepackage{stmaryrd}

\def\O{\mathcal{O}}

\renewcommand{\P}{\mathsf{P}}
\newcommand{\NP}{\mathsf{NP}}
\newcommand{\coNP}{\mathsf{coNP}}
\newcommand{\PSPACE}{\textnormal{\textsf{PSPACE}}}
\newcommand{\IP}{\textnormal{\textsf{IP}}}

\newcommand{\Abs}[1]{\mathopen|#1\mathclose|}

\definecolor{nicebg}{HTML}{f6f0e4}
\definecolor{nicered}{HTML}{7f0a13}
\definecolor{nicebgred}{HTML}{f2e7e8}
\definecolor{niceblue}{HTML}{104354}
\definecolor{nicebgblue}{HTML}{e8edee}
\definecolor{nicegreen}{HTML}{217516}
\definecolor{nicebggreen}{HTML}{e9f1e8}
\definecolor{nicepurple}{HTML}{884bab}
\definecolor{nicebgpurple}{HTML}{f3edf7}
\definecolor{niceorange}{HTML}{d27c11}
\definecolor{nicebgorange}{HTML}{fbf2e8}
\definecolor{nicepink}{HTML}{e95f9f}
\definecolor{nicebgpink}{HTML}{fdeff6}
\definecolor{niceredlight}{HTML}{c9888d}
\definecolor{nicebluelight}{HTML}{78a4b8}
\definecolor{nicegreenlight}{HTML}{76de68}
\definecolor{nicepurplelight}{HTML}{bc87db}
\definecolor{niceredbright}{HTML}{bd0310}
\definecolor{nicebgredbright}{HTML}{f9e6e8}
\definecolor{nicebluebright}{HTML}{197b9b}
\definecolor{nicebgbluebright}{HTML}{e8f2f5}

\newtheorem{definition}  {Definition}
\newtheorem{proposition} {Proposition}
\newtheorem{lemma}       {Lemma}
\newtheorem{example}     {Example}

\makeatletter
\RequirePackage[bookmarks,unicode,colorlinks=true]{hyperref}\def\@citecolor{niceblue}\def\@urlcolor{niceblue}\def\@linkcolor{nicered}
\def\orcidID#1{\smash{\href{http://orcid.org/#1}{\protect\raisebox{1pt}{\protect\includegraphics{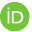}}}}}
\makeatother

\usepackage{etoolbox}
\makeatletter
\pretocmd\start@gather{\if@minipage\kern-\topskip\kern-\baselineskip\kern+7pt\fi
}{}{}
\makeatother

\newcommand{\Parag}[1]{\paragraph*{#1.}}
\newcommand{\parag}[1]{\paragraph*{#1.}}
\newcommand{\proofparag}[1]{\smallskip\noindent\textbf{\textsf{#1.}}}
\newcommand{\Qed}{}

\makeatletter
\renewcommand\paragraph{\scr@startsection{paragraph}{4}{\z@}{.5ex \@plus.25ex \@minus.25ex}{-1em}{\normalfont\normalsize\sffamily\bfseries}}
\renewcommand\subparagraph{\scr@startsection{subparagraph}{5}{\parindent}{0ex \@plus.25ex}{-1em}{\normalfont\normalsize\sffamily\bfseries}}
\makeatother

\usepackage{enumitem}
\setlist[enumerate,1]{itemsep=0pt,topsep=1ex,before={\pagebreak[1]}}
\setlist[itemize,1]{itemsep=0pt,topsep=1ex}

\usepackage[format=plain]{caption}

\newcommand{\appref}[1]{Appendix~\ref{#1}}

\def\Fullversion{1}
\def\set#1{{\{ #1 \}}}
\newcommand{\F}{\mathbb{F}}
\renewcommand{\O}{\mathcal{O}}

\newcommand{\Zero}{\mathbf{0}}
\newcommand{\One}{\mathbf{1}}

\newcommand{\sem}[1]{[\![ #1 ]\!]}
\newcommand{\Eval}[1]{\sem{{#1}}}
\newcommand{\pev}[1]{\pi_{[{#1}]}}
\newcommand{\Pev}[1]{{\mathrm{\Pi}}_{{#1}}}
\newcommand{\Bdd}[1]{\left\langle#1\right\rangle}
\newcommand{\Level}{\operatorname{\ell}}
\newcommand{\Gets}{:=}
\newcommand{\CPEnofmt}{CP}
\newcommand{\CPE}{\textrm{\CPEnofmt}}
\newcommand{\CPEDG}{\textrm{CPD}}
\newcommand{\CPEs}{\CPE{}s}

\newcommand{\bddop}{\operatorname{Apply}_\circledast}
\newcommand{\bddopvee}{\operatorname{Apply}_\vee}

\newcommand{\Claimset}{\mathcal{C}}

\newcommand{\Alg}[1]{\textnormal{\textsc{#1}}}
\newcommand{\ComputeEBDD}{\Alg{ComputeEBDD}}
\newcommand{\EvaluateEBDD}{\Alg{EvaluateEBDD}}
\newcommand{\CPCertify}{\Alg{CPCertify}}
\newcommand{\BDDSolver}{\Alg{BDDSolver}}
\newcommand{\Sumcheck}{\Alg{SumCheck}}

\newcommand{\Solver}[1]{\textnormal{\textsf{#1}}}
\newcommand{\blic}{\Solver{blic}}
\newcommand{\pgbddq}{\Solver{PGBDDQ}}
\newcommand{\caqe}{\Solver{CAQE}}
\newcommand{\depqbf}{\Solver{DepQBF}}
\newcommand{\bloqqer}{\Solver{Bloqqer}}
\newcommand{\qchecker}{\Solver{qchecker}}
\newcommand{\qrpcheck}{\Solver{QRPcheck}}

\newcommand{\True}{\textnormal{\textcolor{niceblue}{\textsf{true}}}}
\newcommand{\False}{\textnormal{\textcolor{niceblue}{\textsf{false}}}}

\renewcommand{\ldots}{...}

\tikzstyle{cnode}=[ellipse,draw,thick,minimum width=2em,minimum height=3ex,fill=nicebgblue]
\tikzstyle{dnode}=[circle,draw,thick,minimum size= 0pt,inner sep=1pt, outer sep=0pt, fill=nicebgredbright]
\tikzstyle{pnode}=[ellipse,draw,thick,fill=nicebgblue]
\tikzstyle{fml}=[rectangle,draw=nicegreen!50!white,fill=nicebggreen,rounded corners,inner sep=1.5mm]
\tikzstyle{poly}=[inner sep=0mm]

\newsavebox{\ProblemBox}
\newlength{\ProblemLength}
\newcommand{\Problem}[4][]{
  \begingroup
  \def\Arg{#1}
  \ifx\Arg\empty\def\Arg{#2}\fi
  \begin{center}
\savebox{\ProblemBox}{\textsf{#2}}\settowidth{\ProblemLength}{\usebox{\ProblemBox}}\usebox{\ProblemBox}\hspace{1mm}
    \begin{tabular}{ll}
      \textbf{Input}& #3\\
      \textbf{Output}& #4\\
    \end{tabular}
  \end{center}
  \endgroup
}
\newcommand{\Pref}[1]{\textnormal{#1}}

\makeatletter
\newcommand{\Customlabel}[2]{\protected@write \@auxout {}{\string \newlabel {#1}{{#2}{\thepage}{#2}{#1}{}} }\hypertarget{#1}{}
}
\makeatother

\begin{document}
\title{Making $\IP=\PSPACE$ Practical: Efficient Interactive Protocols for BDD Algorithms}
\author{Eszter Couillard$^1$ \orcidID{0009-0005-3609-1738}, 
Philipp Czerner$^1$ \orcidID{0000-0002-1786-9592},\\
Javier Esparza$^1$ \orcidID{0000-0001-9862-4919},
Rupak Majumdar$^2$ \orcidID{0000-0003-2136-0542} }
\affil{\large \{couillar, czerner, esparza\}@in.tum.de, rupak@mpi-sws.org\\
$^1$ Department of Informatics, TU München, Germany\\
$^2$ Max Planck Institute for Software Systems, Germany}
\date{}

\maketitle

\vspace*{-9mm}

\begin{abstract}
\textbf{Abstract.}
We show that interactive protocols between a prover and a verifier, a well-known tool of complexity theory, can be used in practice to certify the correctness of automated reasoning tools. 

\smallskip

Theoretically, interactive protocols exist for all $\textsf{PSPACE}$ problems. The verifier of a protocol checks the prover's answer to a problem instance in probabilistic polynomial time, with polynomially many bits of communication, and with exponentially small probability of error.  (The prover may need exponential time.) Existing interactive protocols are not used in practice because their provers use naive algorithms, inefficient even for small instances, that are incompatible with practical implementations of automated reasoning.

\smallskip

We bridge the gap between theory and practice by means of an interactive protocol whose prover uses BDDs. We consider the problem of counting the number of assignments to a QBF instance ($\#\textrm{CP}$), which has a natural BDD-based algorithm. We give an interactive protocol for $\#\textrm{CP}$ whose prover is implemented on top of an extended BDD library. The prover has only a linear overhead in computation time over the natural algorithm.

\smallskip

We have implemented our protocol in $\textsf{blic}$, a certifying tool for $\#\textrm{CP}$. Experiments on standard QBF benchmarks show that $\textsf{blic}$ is competitive with state-of-the-art QBF-solvers. The run time of the verifier is negligible.
While loss of absolute certainty can be concerning, the error probability in our experiments is at most $10^{-10}$ and reduces to $10^{-10k}$ by repeating the verification $k$ times.
 \end{abstract}

\section{Introduction}

Automated reasoning tools often underlie our assertions about the correctness of critical hardware and software components. 
In recent years, the scope and scalability of these techniques have grown significantly.

Automated reasoning tools
are not immune to bugs. If we are to trust their verdict, it is important that they provide evidence of their correct behaviour.
A substantial amount of research has gone into proof-producing automated reasoning tools \cite{Heule21,NeculaPCC-POPL97,Namjoshi01,CAV02BLAST,BarbosaRKLNNOPV22}. 
These works define a notion of ``correctness certificate'' suitable for the reasoning problem at hand, and adapt the reasoning engine to produce independently checkable certificates.
For example, SAT solvers produce either a satisfying assignment or a proof of unsatisfiability in some proof system, e.g.\ resolution (see \cite{Heule21} for a survey).
Extending such certificates beyond boolean SAT is an active area of current research \cite{BarbosaRKLNNOPV22,KatzBTRH16,NiemetzPLSB12,CAQE,BalabanovWJ14}.

In the worst case, the size of certificates grows exponentially in the size of the input, even for boolean unsatisfiability (unless $\NP$ $=$ $\coNP$). If users have limited computational or communication resources, transferring and checking large certificates becomes a burden. Large certificates are not just a theoretical curiosity. 
In practice, resolution proofs for complex SAT problems may run to petabytes \cite{Heule17}.
Ideally, we would prefer ``small'' certificates (polynomial in the size of the input) which can be checked independently in polynomial time.

The  $\IP=\PSPACE$ theorem proves that certification with polynomial
verification time is possible for any problem in $\PSPACE$, provided one trades off absolute
certainty for certainty with high probability \cite{Shamir92}.
The complexity class $\IP$ consists of those languages for which there is a polynomial-round, complete and sound \emph{interactive protocol} 
\cite{GoldwasserMicaliRackoff,Babai,LundFKN92,AroraBarak}---a sequence of interactions between a (computationally unbounded) prover and a (computationally bounded) verifier after which the verifier decides whether the prover correctly performed a computation. The protocol is complete if, whenever an input belongs to the language, 
there is an \emph{honest prover} who can convince a polynomial-time randomised verifier in a polynomial number of rounds.
The protocol is sound if, whenever an input does not belong to the language, the Verifier will reject the input with high probability — no matter what certificates are provided to the Verifier. That is, a “Prover" cannot fool the certification process.

Since every language in $\PSPACE$ has an interactive protocol, there are interactive protocols for UNSAT, QBF,  \emph{counting} QBF,  safety verification of concurrent state machines, etc. Observe that the prover of a protocol may perform exponential time computations (which is unavoidable unless $\P = \PSPACE$), but the verifier only requires polynomial time in the original input.

If interactive protocols provide a foundation for small and efficiently verifiable certificates (at least for problems in $\PSPACE$), why are they not in widespread practice? We believe the reason to be the following: for asymptotic complexity purposes, it suffices to use honest provers with best-case exponential complexity that naively enumerate all possibilities. Such provers are incompatible with automated reasoning tools, which use more sophisticated data structures and heuristics to scale to real-world examples.
So we need to make \emph{practical algorithms} for automated reasoning \emph{efficiently certifying}. We call an algorithm \emph{efficiently certifying} if, in addition to computing the output, it can execute the steps of an honest prover in an interactive protocol with only polynomial overhead over its running time.

In this paper, we show that algorithms using reduced ordered binary decision diagrams 
(henceforth called BDDs) \cite{Bryant86} can be made efficiently certifying.
We consider $\Pref{\#\CPE}$, the problem of computing the number of satisfying assignments of a
\emph{circuit with partial evaluation (\CPE)}. Besides boolean nodes, 
a {\CPE} contains \emph{partial evaluation} nodes $\pev{x:=\False}$ (resp., $\pev{x:=\True}$) that take a boolean predicate as input, say $\varphi$, and output the result of setting $x$ to $\False$ (resp., $\True$) in $\varphi$. $\Pref{\#\CPE}$ generalises SAT, QBF, and \emph{counting} SAT ($\#$SAT), and has a natural algorithm using BDDs: 
Compute BDDs for each node of the circuit in topological order, and count the 
accepting paths of the final BDD.

The theoretical part of the paper proceeds in two steps. First, we present \CPCertify{}, a complete and sound interactive protocol  for $\#\CPE$. 
\CPCertify{} is similar to the \Sumcheck\ protocol \cite{LundFKN92}.
It involves encoding boolean formulas as polynomials over a finite field.
The prover is responsible for producing certain polynomials from the original circuit
and evaluating them at points of the field chosen by the verifier. These polynomials are either multilinear (all exponents are at most $1$) or quadratic (at most $2$).

Second, we show that an honest prover in \CPCertify{} can be implemented on top of a suitably extended BDD library.
The run times of the certifying BDD algorithms are only a constant overhead over the computation time without certification---they depend linearly on the total number of nodes of the intermediate BDDs computed by the prover to solve the $\#\CPE$ instance. 
We use two key insights. The first is an encoding of multilinear polynomials as BDDs; we show that the 
intermediate BDDs represent all the multilinear polynomials a prover needs during the run of \CPCertify{}. The second shows that the quadratic polynomials correspond to \emph{intermediate steps} during the computation of the intermediate BDDs. We extend BDDs with additional ``book-keeping'' nodes that allow the prover to also compute the quadratic polynomials while solving the problem. So computing the polynomials required by \CPCertify{} has \emph{zero} additional cost; the only overhead is the cost of evaluating the polynomials at the field points chosen by the verifier.

We have implemented a certifying $\#\CPE$ solver based on our extended BDD library.
Our experiments show that the solver is competitive with state-of-the-art non-certifying QBF solvers,  and can outperform certifying QBF solvers based on BDDs. The number of bytes exchanged between the prover and the verifier
are an order of magnitude smaller, and Verifier's run time several orders of magnitude smaller, than current encodings of QBF proofs, while bounding the error probability to below $10^{-10}$.
Thus, our results open the way for practically efficient, probabilistic certification of automated reasoning problems using interactive protocols.

\parag{Additional Related Work}
Proof systems for SAT and QBF remain an active area of research---both in theoretical proof complexity and in practical tool development. 
Jussila, Sinz, and Biere \cite{JussilaSB06,SinzB06} showed how to extract extended resolution proofs from BDD operations.
This is the basis for proof-producing SAT and QBF solvers based on BDDs \cite{BryantH21sat,BryantH21qbf,BryantBH22}.
As in our work, the proof uses intermediate nodes produced in the construction of the BDD operations.
We focus on interactive certification instead of extended resolution proofs, which can be exponentially larger than the input formula.

Recently, Luo et al.~\cite{LuoAHPTW22} consider the problem of providing \emph{zero-knowledge} proofs of unsatisfiability, a motivation similar but not equal to ours. Their techniques require the verifier to work in time polynomial in the proof, which can be exponentially bigger than the input formula. In contrast, the verifier of \CPCertify{} runs in polynomial time in the input. Since any language in PSPACE has a zero knowledge proof \cite{Ben-OrGGHKMR88}, our protocol can in principle be made zero knowledge. Whether that system scales in practice is left for future work.

\ifx\Fullversion\undefined
\parag{Full Version} Detailed proofs can be found in the full version of the paper~\cite{fullversion}.
\fi

\section{Preliminaries}

\parag{The Class $\IP$}
An \emph{interactive protocol} between a \emph{Prover} and a \emph{Verifier} consists of a sequence of
interactions in which a Verifier asks questions to a Prover,
receives responses to the questions, and must ultimately decide if a common input $x$ belongs to a language.
The computational power of the Prover is unbounded but the Verifier is a randomised, polynomial-time algorithm. 

\newcommand{\Out}{\operatorname{out}}
Formally, let $P,V$ denote (deterministic) Turing machines.

We say that $(r;m_1,...,m_{2k})$ is a \emph{$k$-round interaction}, with $r,m_1,...,m_{2k}\in\{0,1\}^*$, if $m_{i+1}=V(r,m_1,...,m_i)$ for even $i$ and $m_{i+1}=P(m_1,...,m_i)$ for odd $i$.
We think of $r$ as an additional sequence of bits given to Verifier $V$ that is chosen randomly.
The \emph{output} $\Out(P, V)(x,r,k)$ is defined as $m_{2k}$, where $(r;m_1,...,m_{2k})$ is the unique $k$-round interaction with $m_1=x$.

A language $L$ belongs to $\IP$ if there exist some $V,P_H$ and polynomials $p_1,p_2,p_3$, s.t.\ $V(r,x,m_2,...,m_i)$ runs in time $p_1(\Abs{x})$ for all $r,x,m_2,...,m_i$, and, for each $x$ and an $r\in\{0,1\}^{p_2(\Abs{x})}$ chosen uniformly at random:
\begin{enumerate}
\item (\emph{Completeness}) $x\in L$ implies $\Out(P_H, V)(x,r,p_3(\Abs{x})) = 1$ with probability 1, and
\item (\emph{Soundness}) $x \notin L$ implies that for all $P$ we have $\Out(P, V)(x,r,p_3(\Abs{x})) = 1$ with probability at most $2^{-\Abs{x}}$.
\end{enumerate}

Intuitively, in an interactive protocol, a computationally unbounded Prover interacts with a randomised polynomial-time Verifier for $k$ rounds.
In each round, Verifier sends probabilistic “challenges” to Prover, based on the input and the answers to prior challenges, 
and receives answers from Prover.
At the end of $k$ rounds, Verifier decides to accept or reject the input.
The completeness property ensures that if the input belongs to the language $L$, then there is an “honest” Prover $P_H$ who can always convince Verifier that indeed $x\in L$. If the input does not belong to the language, then the soundness property ensures that Verifier rejects the input with high probability no matter how a (dishonest) Prover tries to convince them.

It is known that $\IP = \PSPACE$ \cite{LundFKN92,Shamir92}, that is, every language in $\PSPACE$ has a polynomial-round interactive protocol.
The proof exhibits an interactive protocol for the language QBF of true quantified boolean formulae; in particular, 
the honest Prover is a polynomial space, exponential time algorithm that uses a truth table representation of the formula to implement the protocol.

\parag{Polynomials}
Interactive protocols make extensive use of polynomials over some prime finite field $\F$.

Let $X$ be a finite set of variables. We use $x, y, z, \ldots$ for variables and $p, q, \ldots$ for polynomials. When we write a polynomial explicitly, we write it in brackets, e.g.\ $[3xy-z^2]$.  
We write $\One$ and $\Zero$ for the polynomials $[1]$ and $[0]$, respectively.
We use the following operations on polynomials:
\begin{itemize}
\item \emph{Sum, difference, and  product}. Denoted $p+q$, $p-q$, $p \cdot q$, and defined as usual. For example, $[3xy-z^2]+[z^2+yz]=[3xy+yz]$ and $[x+y]\cdot[x-y]=[x^2-y^2]$.
\item \emph{Partial evaluation}. Denoted $\pev{x \Gets  a} \, p$, it returns the result of setting variable $x$ to the field element $a$ in the polynomial $p$, e.g.\ $\pev{x \Gets  5}[3xy-z^2]=[15y-z^2]$.
\item \emph{Degree reduction}. Denoted $\delta_x \, p$.  It reduces the degree of $x$ in all monomials of the polynomial to $1$. For example, 
$\delta_x[x^3y+3x^2+7z^2]=[xy+3x+7z^2]$.
\end{itemize}

A \emph{(partial) assignment} is a (partial) mapping $\sigma: X \rightarrow \F$. We write $\Pev\sigma \, p$ for $\pev{x_1:=\sigma(x_1)}...\pev{x_k:=\sigma(x_k)} \, p$, where $x_1,...,x_k$ are the variables for which $\sigma$ is defined. Additionally, we call $\sigma$ \emph{binary} if $\sigma(x)\in\{0,1\}$ for each $x\in X$.

\parag{Binary and multilinear polynomials} A polynomial is \emph{multilinear in $x$} if the degree of $x$ in $p$ is $0$ or $1$. A polynomial is \emph{multilinear} if it is multilinear in all its variables. For example, $[xy-y^2]$ is multilinear in $x$ but not in $y$, and $[3xy -2zy]$ is multilinear. A polynomial $p$ is \emph{binary} if $\Pev\sigma \, p \in \{\Zero,\One\}$ for every binary assignment $\sigma$. Two polynomials $p,q$ are \emph{binary equivalent}, denoted  $p\equiv_b q$, if $\Pev\sigma \, p=\Pev\sigma \, q$ for every binary assignment $\sigma$. 
(Note that non-binary polynomials can be binary equivalent.)

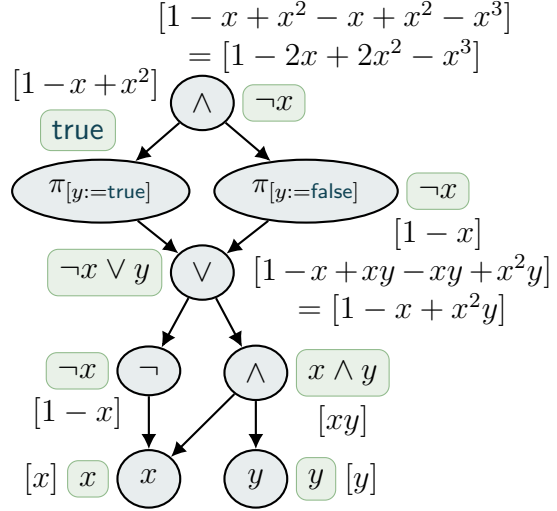
\begin{wrapfigure}{R}{0.5\linewidth}
\tikzstyle{edgeb}=[thick]
\tikzstyle{edger}=[thick]
\tikzset{-{Latex[length=2mm, width=2mm]},node distance=16mm}
\renewcommand{\c}[1]{{\hspace{2.2pt}#1\hspace{1pt}}}
\centering
\begin{tikzpicture}
\node[cnode] (q1) {$\wedge$};
\node[right = 1mm of q1,fml] (p1) {$\neg x$};
\node[above = 1mm of p1,poly,xshift=8mm] (p2) {$\begin{array}{c}[1-x+x^2-x+x^2-x^3]\\=[1-2x+2x^2-x^3]\end{array}$};
\node[cnode, below left = 0.6cm and 0.2 cm of q1] (q2) {\small $\pev{y \Gets \True}$};
\node[above = 1mm of q2,fml,xshift=-3mm] (p2) {$\True$};
\node[above = 1mm of p2,poly,xshift=1mm] (p2b) {$[1\c{-}x\c{+}x^2]$};
\node[cnode, below right = 0.6cm and 0.2cm of q1] (q3) {\small $\pev{y \Gets \False}$};
\node[right = 1mm of q3,fml] (p3) {$\neg x$};
\node[below = 1mm of p3,poly] (p3b) {$[1-x]$};
\node[cnode, below = 15mm of q1] (q4) {$\vee$};
\node[left = 1mm of q4,fml] (p4) {$\neg x \vee y$};
\node[right = 0.5mm of q4,yshift=-2.5mm,poly] (p4b) {$\begin{array}{c}[1\c{-}x\c{+}xy\c{-}xy\c{+}x^2y]\\=[1-x+x^2y]\end{array}$};
\node[cnode, below right= 0.8cm and 1mm of q4] (q5) {$\wedge$};
\node[cnode, below left = 0.8cm and 1mm of q4] (q5b) {$\neg$};
\node[left = 1mm of q5b,fml] (p5b) {$\neg x$};
\node[below = 1mm of p5b,poly] (p5bb) {$[1-x]$};
\node[right = 1mm of q5,fml] (p5) {$x \wedge y$};
\node[below = 1mm of p5,poly] (p5bb) {$[xy]$};
\node[cnode, below = 7mm of q5b] (q6) {$x\vphantom{y}$};
\node[left = 1mm of q6,fml] (p6) {$x$};
\node[left = 1mm of p6,poly] (p6b) {$[x]$};
\node[cnode] (q7) at (q6-|q5) {$y$};
\node[right = 1mm of q7,fml] (p7) {$y$};
\node[right = 1mm of p7,poly] (p7b) {$[y]$};
\draw (q1) edge[edgeb] (q2)
(q1) edge[edger] (q3)
(q2) edge[edgeb] (q4)
(q3) edge[edgeb] (q4)
(q4) edge[edger] (q5)
(q4) edge[edgeb] (q5b)
(q5b) edge[edgeb] (q6)
(q5) edge[edgeb] (q6)
(q5) edge[edger] (q7);
\end{tikzpicture}
\caption{A \CPE\ (Section~\ref{sec:cpe}), the boolean functions represented by each node (in boxes), and the arithmetisation of the formulae (Section~\ref{subsec:arith}).}
\label{fig:circuit}
\vspace*{-10mm}
\end{wrapfigure}

\section{Circuits with Partial Evaluation}
\label{sec:cpe}

We introduce circuits with partial evaluation (\CPE), a compact representation of quantified boolean formulae, 
and formulate \Pref{\#\CPE}, the problem of counting the number of satisfying assignments of a $\CPE$.
\Pref{\#\CPE}  generalises QBF, the satisfiability problem for quantified boolean formulas.
Figure \ref{fig:circuit} shows an example of a \CPE. Informally, it is
a directed acyclic graph whose nodes are labelled with variables, boolean operators, or \emph{partial evaluation operators} $\pev{x \Gets b}$. Intuitively,  $\pev{x \Gets b} \varphi$ sets the variable $x$ to the truth value $b$ in the formula $\varphi$. In this way, each node of a circuit stands for a boolean function, and the complete circuit stands for the boolean function of the root. Figure \ref{fig:circuit} shows the formulae represented by each node.

\newcommand{\Free}{\operatorname{free}}
\newcommand{\Sub}{\operatorname{sub}}
\begin{definition}
Let $X$ denote a finite set of \emph{variables} and $S\subseteq X$. A \emph{circuit with partial evaluation and variables in  $S$} ($S$-{\CPE}) has the form
\begin{itemize}
\item $\True$, $\False$, or $x$, where $x\in S$,
\item $\neg \varphi$, $\varphi \wedge \psi$, or $\varphi \vee \psi$, where $\varphi,\psi$ are $S$-{\CPE}s, or
\item $\pev{y \Gets b} \, \varphi$, where $y\in X\setminus S$, $b\in\{\True,\False\}$, and $\varphi$ is an $(S\cup\{y\})$-{\CPE}.
\end{itemize}
The set of  \emph{free variables} of a $S$-{\CPE} $\varphi$ is $\Free(\varphi):=S$. The \emph{children} of a {\CPE} are inductively defined as follows:  $\True$, $\False$, and $x$ have no children;  the children of $\varphi \wedge \psi$ and 
$\varphi \vee \psi$ are $\varphi$ and $\psi$; and the only child of $\neg \varphi$ and $\pev{y \Gets b}\, \varphi$ is $\varphi$.
The set of \emph{descendants} of $\varphi$ is the smallest set $M$ containing $\varphi$ and all children of every element of $M$. The \emph{size} of $\varphi$ is $\Abs{\varphi}:=\Abs{M}$.
\end{definition}

We represent a {\CPE} $\varphi$ as a directed acyclic graph. The nodes of the graph are the descendants of $\varphi$.
A {\CPE} $\varphi$ encodes a boolean predicate $P_\varphi$, which maps assignments $\sigma \colon \Free(\varphi)\rightarrow\{\False,\True\}$ to a truth value $P_\varphi(\sigma)\in\{\False,\True\}$. It does so in the obvious manner, e.g., $P_x(\sigma):=\sigma(x)$, $P_{\varphi\wedge\psi}(\sigma):=P_\varphi(\sigma) \wedge P_\psi(\sigma)$, etc. We use $\pev{x \Gets  b}$ as partial evaluation operator, so $P_{\pev{x \Gets  b}\varphi}(\sigma):=P_\varphi(\sigma\cup\{x\mapsto b\})$. Intuitively, $\pev{x \Gets  b} \, \varphi$ replaces each occurrence of $x$ in $\varphi$ by $b$. An assignment $\sigma$ \emph{satisfies} $\varphi$ if $P_\varphi(\sigma)=\True$.
We define the macros
\begin{align*}
\forall_x\varphi &:= \pev{x \Gets  0} \, \varphi\wedge\pev{x \Gets  1} \, \varphi\\
\exists_x\varphi &:= \pev{x \Gets  0} \, \varphi \vee\pev{x \Gets  1} \, \varphi
\end{align*}
Figure~\ref{fig:circuit} shows a $\CPE$ for the quantified boolean formula $\forall_y(\neg x \vee (x \wedge y))$.

We consider the following problem:
\Problem[CCPE]{\#{\CPEnofmt}}
{{\CPE} $\varphi$.}
{The number of satisfying assignments of $\varphi$.}

Given a quantified boolean formula, we can use the macros for quantifiers to construct in linear time an equivalent {\CPE}, i.e., 
a {\CPE} with the same satisfying assignments. Similarly, \textsf{\#SAT} instances can also be reduced to \Pref{\#\CPE}.

\Parag{Structure of the rest of the paper} 
In Section~\ref{sec:ip-cpe}, we give an interactive protocol for \Pref{\#\CPE} called \CPCertify{}.
In Section~\ref{sec:bddprover}, we implement an honest Prover for \CPCertify{} on top of 
an extended BDD-based algorithm for \Pref{\#\CPE}.
The prover runs in time polynomial in the size of the largest BDD for any of the subcircuits of the initial circuit.
Together, these results yield our main result, Theorem \ref{thm:main}, showing that any BDD-based algorithm can be modified to run an interactive protocol with small polynomial overhead. Finally, Section~\ref{sec:evaluation} presents empirical results.

\section{An Interactive Protocol for \Pref{\#\CPE} }
\label{sec:ip-cpe}

In this section we describe an interactive protocol for \Pref{\#\CPE}, following the \Sumcheck{} protocol of \cite{LundFKN92}.  
Section \ref{subsec:arith} introduces  arithmetisation, a technique to transform \Pref{\#\CPE} into an equivalent problem about polynomials. 
Section \ref{subsec:degree-reduction} shows how to transform \Pref{\#\CPE} into an equivalent problem about
evaluating polynomials \emph{of low degree}. 
Finally, Section \ref{subsec:CPCertify} presents an interactive protocol for this problem. 

\subsection{Arithmetisation}
\label{subsec:arith}
We define a mapping $\Eval{\cdot}$ that assigns to each \CPE{} $\varphi$ a polynomial $\Eval{\varphi}$ over the variables $\Free(\varphi)$, called the \emph{arithmetisation} of $\varphi$:
\begin{itemize}
\item $\Eval{\True} := \One$; $\Eval{\False} := \Zero$; $\Eval{x} := [x]$ for every $x \in X$; and $\Eval{\neg \varphi} := \One - \Eval{\varphi}$;
\item  $\Eval{\varphi\wedge\psi} :=\Eval{\varphi} \cdot \Eval{\psi}$; and $\Eval{\varphi\vee\psi} := \Eval{\varphi}+ \Eval{\psi}-\Eval{\varphi}\cdot\Eval{\psi}$;
\item $\Eval{\pev{x \Gets  b} \, \varphi}:=\pev{x \Gets  \Eval{b}}\Eval{\varphi}$, with $x\in \Free(\varphi)$, $b\in\{\True,\False\}$.
\end{itemize}

Figure \ref{fig:circuit} also shows the polynomials corresponding to the nodes of the {\CPE}.

Let $\F$ be a fixed prime finite field.
Given an arbitrary truth assignment $\sigma \colon X \to \{\True, \False\}$, let $\overline{\sigma} \colon X \to \F$  be the 
binary assignment given by  $\overline{\sigma}(x) = 1$ if $\sigma(x) = \True$ and $\overline{\sigma}(x) = 0$ if $\sigma(x) = \False$, where $0$ and $1$ denote the additive
and multiplicative identities in $\F$. 
The mapping $\Eval{\cdot}$ is defined to satisfy the following property, whose proof is immediate:

\begin{proposition}
\label{prop:fundeval}
Let  $\varphi$ be an $S$-{\CPE} encoding some boolean predicate $P_\varphi$. Then 
$P_\varphi(\sigma) = \Pev{\overline{\sigma}}\, \Eval{\varphi}$ for every truth assignment $\sigma$ to $S$. 
\end{proposition}

So, intuitively, the polynomial $\Eval{\varphi}$ is a conservative extension of the predicate $P_\varphi$: It returns the same values for all binary assignments. 
Accordingly, in the rest of the paper we abuse language and write $\sigma$ instead of $\overline{\sigma}$ for the binary assignment corresponding to the truth assignment $\sigma$.

Observe that \Pref{\#\CPE} can be reformulated as follows: given a \CPE\ $\varphi$, compute the number of binary assignments $\sigma$ s.t.\ $\Pev{\sigma}\Eval{\varphi} = \One$. 

\subsection{Degree Reduction}
\label{subsec:degree-reduction}

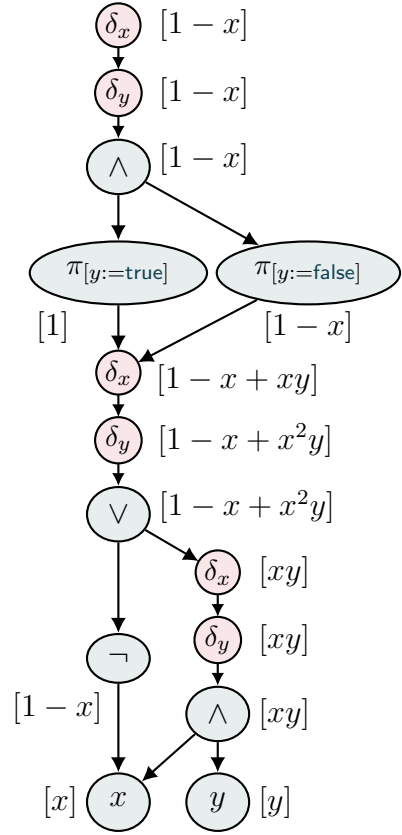
\begin{wrapfigure}{R}{0.357\linewidth}
\tikzstyle{edgeb}=[thick]
\tikzstyle{edger}=[thick]
\tikzset{-{Latex[length=2mm, width=2mm]},node distance=16mm}
\begin{tikzpicture}
\node[dnode] (q11) {\small $\delta_x$};
\node[dnode, below = 3mm of q11] (q12) {\small $\delta_y$};
\node[cnode, below = 3mm of q12] (q1) {$\wedge$};
\node[cnode, below = 6mm of q1] (q2) {\small $\pev{y \Gets \True}$};
\node[cnode, right = 1mm of q2] (q3) {\small $\pev{y \Gets \False}$};
\node[dnode, below = 6mm of q2] (q41) {\small $\delta_x$};
\node[dnode, below = 3mm of q41] (q42) {\small $\delta_y$};
\node[cnode, below = 3mm of q42] (q4) {$\vee$};
\node[dnode, below right = 3mm and 8mm of q4] (q51) {\small $\delta_x$};
\node[dnode, below = 3mm of q51] (q52) {\small $\delta_y$};
\node[cnode, below = 3mm of q52] (q5) {$\wedge$};
\node[cnode, below = 12mm of q4] (q5b) {$\neg$};
\node[cnode, below = 4mm of q5] (q7) {$y$};
\node[cnode] (q6) at (q7-|q5b)  {$x\vphantom{y}$};
\node[right = 1mm of q7,poly] (p7) {$[y]$};
\node[left  = 1mm of q6,poly] (p7) {$[x]$};
\node[below left = 2mm and -1mm of q5b,poly] (p5b) {$[1-x]$};
\node[right = 1mm of q5,poly] (p5) {$[xy]$};
\node[poly] (p51) at (q51-|p5) {$[xy]$};
\node[poly] (p52) at (q52-|p5) {$[xy]$};
\node[right = 1mm of q4,poly,yshift=1mm] (p4) {$[1-x+x^2y]$};
\node[right = 2mm of q42,poly] (p42) {$[1-x+x^2y]$};
\node[right = 2mm of q41,poly,yshift=-1mm] (p41) {$[1-x+xy]$};
\node[below left = 2mm and -2mm of q2,poly] (p2) {$[1]$};
\node[right = 1mm of q1,poly,yshift=1mm] (p1) {$[1-x]$};
\node[poly] (p12) at (q12-|p1) {$[1-x]$};
\node[poly] (p11) at (q11-|p1) {$[1-x]$};
\node[poly] (p3) at (p2-|q3) {$[1-x]$};
\draw (q1) edge[edgeb] (q2)
(q1) edge[edger] (q3)
(q2) edge[edgeb] (q41)
(q3) edge[edgeb] (q41)
(q4) edge[edger] (q51)
(q4) edge[edgeb] (q5b)
(q5b) edge[edgeb] (q6)
(q5) edge[edgeb] (q6)
(q5) edge[edger] (q7);
\draw[-{Latex[length=1.7mm, width=1.7mm]}]
(q11) edge[edgeb] (q12)
(q12) edge[edgeb] (q1)
(q41) edge[edgeb] (q42)
(q42) edge[edgeb] (q4)
(q51) edge[edgeb] (q52)
(q52) edge[edgeb] (q5);
\end{tikzpicture}
\caption{\CPEDG\ and polynomials for the \CPE\ of Figure~\ref{fig:circuit}.}
\label{fig:red-polynomials}
\vspace*{-28mm}
\end{wrapfigure}
Given a {\CPE} $\varphi$, its associated polynomial can have degree exponential in the height of $\varphi$. 
Since we are ultimately interested in evaluating polynomials over binary assignments, and since $x^2=x$ for $x\in\set{0,1}$,
we can convert polynomials to low degree without changing their behaviour on binary assignments. 

For this, we use a \emph{degree-reduction} operator $\delta_x$ for every variable $x$. 
The operator $\delta_x p$ reduces the exponent of all powers of $x$ in $p$ to $1$. 
For example, $\delta_x [x^2y + 3 x y^2 - 2x^3 y^2+ 4] = [xy + 3x y^2 - 2xy^2 + 4]$.
Observe that $\delta_x p\equiv_b p$.
Instead of working on the input {\CPE} directly, we first convert it into a \emph{circuit with partial evaluation and degree reduction} by inserting 
degree-reduction operators after binary operations. 
This ensures all intermediate polynomials obtained by arithmetisation have low degree.

\begin{definition}
A \emph{circuit with partial evaluation and degree reduction over the set $S$ of variables ($S$-{\CPEDG})} is defined in the same manner as an $S$-{\CPE}, extended as follows:
\begin{itemize}
\item if $\varphi$ is an $S$-{\CPEDG}  and $x\in S$, then $\delta_x\varphi$ is an $S$-{\CPEDG},
\item $\Eval{\delta_x\varphi}:=\delta_x\Eval{\varphi}$, and
\item $\varphi$ is the only child of $\delta_x\varphi$.
\end{itemize}
For an $S$-{\CPEDG} $\varphi$ we define $\Free(\varphi)$, $\Abs{\varphi}$, children, descendants, and the graphical representation as for $S$-{\CPE}s.
\end{definition}

\newcommand{\Convert}{\operatorname{conv}}
We convert a {\CPE} $\varphi$ into a {\CPEDG} $\Convert(\varphi)$ by adding a degree-reduction operator for each free variable before any binary operation. 

\begin{definition}
Given a {\CPE} $\varphi$ with $\Free(\varphi)=\{x_1,...,x_k\}$, its associated {\CPEDG} $\Convert(\varphi)$ is inductively defined as follows:
\begin{itemize}
\item $\Convert(\False)=\False$, $\Convert(\True):=\True$, 
\item $\Convert(\neg \psi):=\neg\Convert(\psi)$,
\item $\Convert(\pev{x \Gets  b}\, \psi):=\pev{x \Gets  b} \, \Convert(\psi)$, and
\item $\Convert(\psi_1\circledast\psi_2):=\delta_{x_1}...\delta_{x_k}(\Convert(\psi_1) \circledast \Convert(\psi_2))$, for $\circledast\in\{\vee,\wedge\}$.
\end{itemize}
\end{definition}

Figure \ref{fig:red-polynomials} shows the {\CPEDG} $\Convert(\varphi)$ for the {\CPE} $\varphi$ of Figure \ref{fig:circuit}, together with the polynomials corresponding to each node.

We collect some basic properties of {\CPEDG}s:

\begin{restatable}{lemma}{lemconv}\label{lem:conv}
Let $\varphi$ be a {\CPE}. 
\begin{enumerate}[label={(\alph*)}]
\item $\Eval{\Convert(\varphi)}$ is a binary multilinear polynomial and $\Eval{\Convert(\varphi)}\equiv_b\Eval{\varphi}$.
\item For every descendant $\psi$ of $\Convert(\varphi)$, $\Eval{\psi}$ has maximum degree $2$.
\end{enumerate}
\end{restatable}

\CPEDG{}s have another useful property. Recall that given a \CPE\ $\varphi$ we are interested in its number of satisfying assignments.
The next lemma shows that this number can be computed by evaluating the polynomial $\Eval{\Convert(\varphi)}$ on \emph{a single input}.

\begin{restatable}{lemma}{numberofsat}
\label{lem:numberofsat}
A {\CPE} $\varphi$ with $n$ free variables has $m < \Abs{\F}$ satisfying assignments if{}f $\Pev{\sigma}\Eval{\Convert(\varphi)}=m\cdot2^{-n}$, where $\sigma$ is the assignment satisfying $\sigma(x) := 2^{-1}$ in the field $\F$ for every variable $x$.\footnote{Any prime field $\F$ with $\Abs{\F}>2$ has an element $c$ such that  $2 c = 1$.}
\end{restatable}

\subsection{\CPCertify: An Interactive Protocol for \Pref{\#\CPE}}
\label{subsec:CPCertify}

We describe an interactive protocol, called \CPCertify, for a {\CPE} $\varphi$ with $n$ free variables. Let $X$ denote the variables used in $\varphi$.
Prover and Verifier fix a finite field with at least $m+1$ elements, where $m$ is an upper bound on the number of assignments (e.g.\ $m=2^n$).
Prover tries to convince the Verifier that $\Pev\sigma\Eval{\Convert(\varphi)} = K$ for some $K \in \F$. 

In the protocol, Verifier challenges Prover to compute polynomials of the form $\Pev{\sigma}(\Eval{\psi})$, where $\psi$ is a node of the {\CPEDG}
$\Convert(\varphi)$ and $\sigma \colon \Free(\psi)\rightarrow\F$ is a (non-binary!) assignment; we call the expression 
$\Pev{\sigma}\Eval{\Convert(\psi)}$ a \emph{challenge}. Observe that all assignments are chosen by Verifier.
Prover answers with some $k \in \F$. We call the expression
$\Pev{\sigma}\Eval{\Convert(\psi)}=k$ a \emph{claim}, or the \emph{answer} to the challenge $\Pev{\sigma}\Eval{\Convert(\psi)}$.

\CPCertify\ consists of an initialisation and a number of rounds, one for each descendant of $\Convert(\varphi)$. 
Rounds are executed in topological order, starting at the root, i.e.\ at $\Convert(\varphi)$ itself.
The structure of a round for a node $\psi$ of  $\Convert(\varphi)$ depends on whether $\psi$ is an internal node (including the root),
or a leaf. 

At each point, Verifier keeps track of a set $\Claimset $ of claims that must be checked.

\Parag{Initialisation} Verifier sends Prover the challenge $\Pev{\sigma}\Eval{\Convert(\varphi)}$, where  $\sigma(x) := 2^{-1}$ for every $x \in \Free(\varphi)$. 
Prover returns the claim $\Pev\sigma\Eval{\Convert(\varphi)}=K$ for some $K \in \F$. (By Lemma \ref{lem:numberofsat}, this amounts to claiming that $\varphi$ has $K \cdot 2^n$ satisfying assignments.) Verifier initialises $\Claimset :=\{\Pev\sigma\Eval{\Convert(\varphi)}=K\}$.

\Parag{Round for an internal node}  A round for an internal node $\psi$ runs as follows:
\begin{itemize}
\item[(a)] Verifier collects all claims $\{ \Pev{\sigma_i}\Eval{\psi} = k_i\}_{i=1}^m$ in $\Claimset $ relating to $\psi$, with assignments $\sigma_1,\ldots,\sigma_m \colon \Free(\psi)\rightarrow\F$ and  $k_1,...,k_m\in\F$. (Initially $\psi=\Convert(\varphi)$ and the only claim is $\Pev\sigma\Eval{\Convert(\varphi)}=K$.)

\item[(b)] If $m > 1$, Verifier interacts with Prover to compute a unique claim $\Pev{\sigma}\Eval{\psi} = k$ such that very likely\footnote{The precise bound on the failure probability will be given in Proposition~\ref{prop:algorithm}.} the claim is true only if all claims $\{\Pev{\sigma_i}\Eval{\psi} = k_i\}_{i=1}^m$ are true. For this, Verifier sends a number of challenges, and checks that the answers are \emph{consistent} with the prior claims. Based on these answers, Verifier then derives new claims. (See ``Description of step (b)'' below.)

\item[(c)] Verifier interacts with Prover to compute a claim $\Pev{\sigma'}\Eval{\psi'}=k'$ for each child $\psi'$ of $\psi$. This is similar to (b): if $\Pev{\sigma}\Eval{\psi} \ne k$, i.e.\ the unique claim from (b) does not hold, then very likely one of the resulting claims will be wrong. Depending on the type of $\psi$, the claims are computed based on the answers of Prover to challenges sent by Verifier. (See ``Description of step (c)'' below.)

\item[(d)] In total, Verifier removed the claims $\{ \Pev{\sigma_i}\Eval{\psi} = k_i\}_{i=1}^m$ from $\Claimset $, and replaced them by one claim $\Pev{\sigma'}\Eval{\psi'}=k'$ for each child $\psi'$ of $\psi$.
\end{itemize}
Observe that, since a node $\psi$ can be a child of several nodes, Verifier may collect multiple claims for $\psi$, one for each parent node.

\Parag{Round for a leaf}  If $\psi$ is a leaf, then $\psi=x$ for a variable $x$, or $\psi\in\{\True,\False\}$. Verifier removes all claims $\{ \Pev{\sigma_i}\Eval{\psi} = k_i\}_{i=1}^m$ from $\Claimset $, computes the values $c_i:=\Pev{\sigma_i}\Eval{\psi}$, and rejects if $k_i\ne c_i$ for any $i$.

Observe that if all claims made by Prover about leaves are true, then very likely Prover's initial claim is also true. 

\Parag{Description of step (b)} Let $\{ \Pev{\sigma_i}\Eval{\psi} = k_i\}_{i=1}^m$ be the claims in $\Claimset $ relating to node $\psi$. Verifier and Prover conduct step (b) as follows: 
\begin{enumerate}[align=left]
\item[(b.1)] While there exists $x\in X$ s.t.\ $\sigma_1(x), \ldots ,\sigma_m(x)$ are not pairwise equal:
\begin{enumerate}[align=left]
\item[(b.1.1)] For every $i\in\{1,...,m\}$, let $\sigma_i'$ denote the partial assignment which is undefined on $x$ and otherwise matches $\sigma_i$. Verifier sends the challenges 
$\{\Pev{\sigma'_i}\Eval{\psi}\}_{i=1}^m$ to Prover. Prover answers with claims $\{\Pev{\sigma'_i}\Eval{\psi}=p_i\}_{i=1}^m$. Note that $p_1, \ldots, p_m$ are univariate polynomials with free variable $x$.
\item[(b.1.2)] Verifier checks whether $k_i=\pev{x \Gets  \sigma_i(x)} \, p_i$ holds for each $i$. If not, Verifier rejects. Otherwise, Verifier picks $r\in\F$ uniformly at random and updates $\sigma_i(x):=r$ and $k_i:=\pev{x \Gets  r}p_i$ for every $i\in\{1,...,m\}$.
\end{enumerate}
\item[(b.2)] If after exiting the loop the values $k_1,...,k_m$ are not pairwise equal, Verifier rejects. Otherwise (that is, if $k_1 =k_2 = \cdots =k_m)$, the set $\Claimset $ now contains a unique claim $\Pev\sigma\Eval{\psi}=k$ relating to $\psi$.
\end{enumerate}

\begin{example}
Consider the case in which $X = \{x\}$, and Prover has made two claims, $\Pev{\sigma_1}\Eval{\psi} = k_1$ and  $\Pev{\sigma_2}\Eval{\psi} = k_2$ with $\sigma_1(x)=1$ and $\sigma_2(x)=2$. 
In step (b.1.1) we have $\sigma_1'=\sigma_2'$ (both are  the empty assignment), and so Verifier sends the challenge
$\Eval{\psi}$ to Prover twice, who answers with claims $\Eval{\psi}=p_1$ and $\Eval{\psi}=p_2$. In step (b.1.2) Verifier checks that $p_1(1)=k_1$ and $p_2(2)=k_2$ hold, picks a random number $r$,
and updates $\sigma_1(x):=\sigma_2(x):= r$ and $k_1:=p_1(r),k_2:=p_2(r)$. Now the condition of the while loop fails, so Verifier moves to (b.2) and checks $k_1=k_2$.
\end{example}

\Parag{Description of step (c)} Let $\Pev{\sigma}\Eval{\psi} = k$ be the claim computed by Verifier in step (b). Verifier removes this claim from $\Claimset $ and replaces it by claims about the children of $\psi$, depending on the structure of $\psi$: 
\begin{enumerate}[align=left]
\item[(c.1)] If $\psi=\psi_1\circledast\psi_2$, for a $\circledast\in\{\vee,\wedge\}$, then Verifier sends Prover challenges $\Pev\sigma\Eval{\psi_i}$ for $i\in\{1,2\}$, and Prover sends claims $\Pev\sigma\Eval{\psi_i}=k_i$ back. Verifier checks the consistency condition $k=\pev{x \Gets  k_1}\pev{y \Gets k_2}\Eval{x\circledast y}$, rejecting if it does not hold.  If the condition holds, the claim $\Pev\sigma\Eval{\psi_i}=k_i$ is added to $\Claimset $, to be checked in the round for $\psi_i$.

\item[(c.2)] If $\psi=\neg \psi'$, then Verifier adds the claim $\Pev\sigma\Eval{\psi'}=1-k$ to $\psi'$.

\item[(c.3)] If $\psi=\pev{x \Gets  b} \, \psi'$, Verifier sets $\sigma' := \sigma\cup\{x \mapsto b\}$ and adds  the claim $\Pev{\sigma'}\Eval{\psi'}=k$ to $\Claimset $.

\item[(c.4)] If $\psi=\delta_x\psi'$, then Verifier sends Prover the challenge $\Pev{\sigma'}\Eval{\psi'}$, where $\sigma'$ denotes the partial assignment which is undefined on $x$ and otherwise matches $\sigma$. Prover returns the claim $p:=\Pev{\sigma'}\Eval{\psi'}$. Observe that $p$ is a univariate polynomial over $x$. Verifier checks the consistency condition $\pev{x \Gets  \sigma(x)}\delta_x \, p=k$, rejecting if it does not hold. If it holds, Verifier picks an $r\in\F$ uniformly at random, conducts the updates $\sigma(x):=r$ and $k:=\pev{x \Gets  r} \, p$, and adds $\Pev\sigma\Eval{\psi'}=k$ to the set of claims about $\psi'$.
\end{enumerate}

This concludes the description of the interactive protocol. 
We now show \CPCertify\ is complete and sound.

\begin{restatable}[\CPCertify\ is complete and sound]{proposition}{cpcertifycorrect}\label{prop:algorithm}
Let $\varphi$ be a {\CPE} with $n$ free variables.
Let $\Pev\sigma\Eval{\Convert(\varphi)}=K$ be the claim initially sent by Prover to Verifier. 
If the claim is true, then Prover has a strategy to make Verifier accept. 
If not, for every Prover, Verifier accepts with probability at most $4n\Abs{\varphi}/\Abs{\F}$.
\end{restatable}

If the original claim is correct, Prover can answer every challenge truthfully and all claims pass all of Verifier's checks.
So Verifier accepts.
If the claim is not correct, we proceed round by round.
We bound the probability that the Verifier is tricked in a single step to at most $2/\Abs{\F}$ using the Schwartz-Zippel Lemma.
We then bound the number of such steps to $2n \Abs{\varphi}$ and use a union bound.

\section{A BDD-based Prover}
\label{sec:bddprover}

We assume familiarity with \emph{reduced ordered binary decision diagrams} (BDDs) \cite{Bryant86}.
We use BDDs over $X=\{x_1, \ldots, x_n\}$. We fix the variable order $x_1 < x_{2} < \ldots < x_n$, i.e.\ the root node would decide based on the value of $x_n$.\footnote{Throughout the paper we use the convention that $x_1$ is near the leaves, and $x_n$ at the root.}

\begin{definition}
\label{def:BDD}
BDDs are defined inductively as follows: 
\begin{itemize}
\item $\Bdd{\True}$ and $\Bdd{\False}$ are BDDs of level $0$;
\item if $u\ne v$ are BDDs of level $\ell_u, \ell_v$ and $i > \ell_u, \ell_v$, then $\Bdd{x_i, u, v }$ is a BDD of level $i$;
\item we identify $\Bdd{x_i, u, u }$ and  $u$, for a BDD $u$ of level $\ell_i$ and $i>\ell_u$.
\end{itemize}
The level of a BDD $w$ is denoted $\Level(w)$. 
The set of \emph{descendants} of $w$ is the smallest set $S$ with $w\in S$ and $u,v\in S$ for all $\Bdd{x,u,v}\in S$. 
The \emph{size} $\Abs{w}$ of $w$ is the number of its descendants. 

The \emph{arithmetisation} of a BDD $w$ is the polynomial $\Eval{w}$ defined
as follows: $\Eval{\Bdd{\True}}:=\One$, $\Eval{\Bdd{\False}}:=\Zero$ and $\Eval{\Bdd{x,u,v}}:=[1-x]\cdot\Eval{u}+[x]\cdot\Eval{v}$.
\end{definition}

Figure \ref{fig:BDD} shows a BDD for the boolean function $\varphi(x,y,z)=(x \wedge y \wedge \neg z) \vee (\neg x  \wedge y \wedge z ) \vee (x \wedge \neg y \wedge z)$
and the arithmetisation of each node.

\parag{\BDDSolver: A BDD-based Algorithm for \textnormal{\Pref{\#\CPE}}}
An instance $\varphi$ of \Pref{\#\CPE} can be solved using BDDs.
Starting at the leaves of $\varphi$, we iteratively compute a BDD for each node $\psi$ of the circuit
encoding the boolean predicate $P_\psi$. 
At the end of this procedure we obtain a BDD for $P_\varphi$. The number of satisfying assignments of $\psi$ is the number of
accepting paths of the BDD, which can be computed in linear time in the size of the BDD. 

For a node $\psi = \psi_1 \circledast \psi_2$, given BDDs representing the predicates $P_{\varphi_1}$ and $P_{\varphi_2}$, we compute a BDD for the predicate 
$P_\varphi := P_{\varphi_1} \circledast P_{\varphi_2}$, using the $\bddop$ operator on BDDs.
We name this algorithm for solving \Pref{\#\CPE} ``\BDDSolver.''

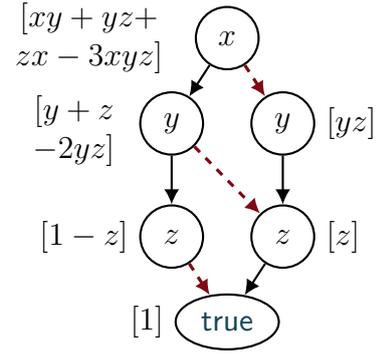
\begin{wrapfigure}{R}{0.37\textwidth}
\tikzstyle{bddnode}=[circle,draw,thick,minimum width=2em,minimum height=3ex]
\tikzstyle{product}=[bddnode,rectangle,draw=niceblue,fill=nicebgblue]
\tikzstyle{edge1}=[thick]
\tikzstyle{edge0}=[edge1,very thick,dashed,nicered]
\tikzstyle{edgep}=[edge1,niceblue]
\tikzset{-{Latex[length=2mm, width=2mm]},node distance=15mm}
\tikzstyle{diag}=[node distance=16mm]
\renewcommand{\boxed}[1]{[#1]}
\centering
\begin{tikzpicture}
\begin{scope}[scale=0.5]
\node[bddnode] (q1) {$x$};
\node[left = 0cm of q1,xshift=-1mm] (l1) {${\begin{array}{c}[xy + yz +\\ zx - 3xyz]\end{array}}$};
\node[bddnode, below left of=q1,diag,xshift=4mm] (q2) {$y$};
\node[left = 0cm of q2,yshift=-1mm] (l2) {${\begin{array}{c}[y+z\\-2yz]\end{array}}$};
\node[bddnode, below right of=q1,diag,xshift=-4mm] (q3) {$y$};
\node[right = 0cm of q3] (l3) {$\boxed{yz}$};
\node[bddnode, below of=q2] (q4) {$z$};
\node[left = 0cm of q4] (l4) {$\boxed{1-z}$};
\node[bddnode, below of=q3] (q5) {$z$};
\node[right = 0cm of q5] (l5) {$\boxed{z}$};
\node[bddnode, below right of=q4,diag,xshift=-4mm,ellipse] (q6) {$\True$};
\node[left = 0cm of q6] (l5) {$\boxed{1}$};  
\draw (q1) edge[edge0] (q3)
(q1) edge[edge1] (q2)
(q2) edge[edge1] (q4)
(q2) edge[edge0] (q5)
(q3) edge[edge1] (q5)
(q4) edge[edge0] (q6)
(q5) edge[edge1] (q6);
\end{scope}
\end{tikzpicture}
\caption{A BDD and its arithmetisation. For $\Bdd{x, u, v}$, we denote the link
from $x$ to $v$ with a solid edge and $x$ to $u$ with a dotted edge. We omit links to $\Bdd{\False}$.
\label{fig:BDD}
}
\ifx\Fullversion\undefined
\vspace*{-10mm}
\else
\vspace*{-3mm}
\fi
\end{wrapfigure}

\parag{From \BDDSolver\ to \CPCertify}
Our goal is to modify \BDDSolver\ to play the role of an honest Prover in \CPCertify\ with minimal overhead.
In \CPCertify, Prover repeatedly performs the same task: 
evaluate polynomials of the form $\Pev\sigma\Eval{\psi}$, where $\psi$ is a descendant of the {\CPEDG} $\Convert(\varphi)$,
and $\sigma$ assigns values to all free variables of $\psi$ except possibly one. 
Therefore, the polynomials have at most one free variable and, as we have seen, degree at most 2. 

Before defining the concepts precisely, we give a brief overview of this section.
\begin{itemize}
\item First (Proposition~\ref{prop:basic}), we show that BDDs correspond to binary multilinear polynomials. In particular, BDDs allow for efficient evaluation of the polynomial. 
As argued in Lemma~\ref{lem:conv}(a), for every descendant $\psi$ of $\varphi$, the CPD $\Convert(\psi)$ (which is a descendant of $\Convert(\varphi)$) evaluates to a multilinear polynomial. In particular, Prover can use standard BDD algorithms to calculate the corresponding polynomials $\Pev\sigma\Eval{\psi}$ for all descendants $\psi$ of $\Convert(\varphi)$ that are neither binary operators nor degree reductions.
\end{itemize}
\par
\ifx\Fullversion\undefined
\vspace{-\baselineskip-\topsep}
\else
\vspace{-1ex}
\fi
\begin{itemize}
\item Second (the rest of the section), we prove a surprising connection: the intermediate results obtained while executing the BDD algorithms 
(with slight adaptations) correspond precisely to the remaining descendants of $\Convert(\varphi)$.
\end{itemize}

The following proposition proves that BDDs represent exactly the binary multilinear polynomials. 

\begin{restatable}{proposition}{basicproperties}\label{prop:basic}
\begin{enumerate}[label={(\alph*)}]
\item For a BDD $w$, $\Eval{w}$ is a binary multilinear polynomial.
\item For a binary multilinear polynomial $p$ there is a unique BDD $w$ s.t.\ $p =\Eval{w}$.
\end{enumerate} 
\end{restatable}
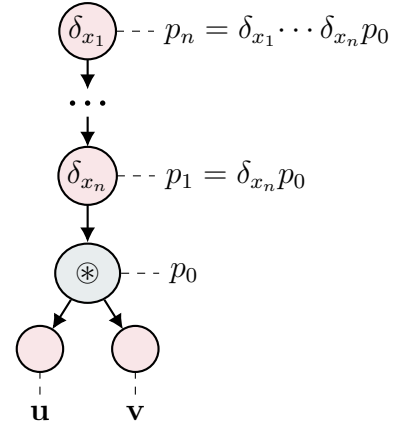
\begin{wrapfigure}{R}{0.4\linewidth}
\tikzstyle{edge1}=[thick]
\begin{center}
\begin{tikzpicture}
\begin{scope}[scale=0.8, xshift=6cm, yshift=-2cm]
\node[cnode] (n) {$\circledast$};
\node[right = 0.5cm of n](ln){$p_0$};
\node[dnode, above = 0.5cm of n](dy){$\delta_{x_n}$};
\node[right = 0.5cm of dy](ldy){$p_1 = \delta_{x_n} p_0$};
\node[above = 0.4cm of dy](dots){\LARGE $\ldots$};
\node[dnode, above = 0.4cm of dots](dx){$\delta_{x_1}$};
\node[right = 0.5cm of dx](ldx){$p_n = \delta_{x_1} \!\cdots \delta_{x_n} p_0$};
\node[dnode, below left = 0.5cm and 0.1cm of n](n1){\phantom{$\delta_x$}};
\node[below= 0.3 of n1] (l11) {$\mathbf{u}$};
\node[dnode, below right = 0.5cm and 0.1cm of n](n2){\phantom{$\delta_x$}};  
\node[below= 0.3 of n2] (l21) {$\mathbf{v}$};    
\draw[-{Latex[length=2mm, width=2mm]}]
(dx) edge[edge1] (dots)
(dots) edge[edge1] (dy)
(dy) edge[edge1] (n)
(n) edge[edge1] (n1)
(n) edge[edge1] (n2);
\draw
(n) edge[dashed] (ln)
(dx) edge[dashed] (ldx)
(dy) edge[dashed] (ldy)
(l11) edge[dashed] (n1)
(l21) edge[dashed] (n2);
\end{scope}
\end{tikzpicture}
\end{center}
\vspace*{-6mm}
\caption{A node of a {\CPE} ($\circledast$) gets a chain of degree reduction nodes in the associated \CPEDG.}
\label{fig:chain}
\ifx\Fullversion\undefined
\vspace*{-9mm}
\else
\vspace*{-3mm}
\fi
\end{wrapfigure}

\subsection{Extended BDDs}

During the execution of \CPCertify\ for a given {\CPEDG} $\Convert(\varphi)$,
Prover sends to Verifier claims of the form $\Pev\sigma \Eval{\psi}$, where $\psi$ is a descendant of $\Convert(\varphi)$, and $\sigma \colon X \rightarrow \F$ is a partial assignment. 
While all polynomials computed by \CPCertify\ are binary, not all are multilinear: some polynomials have degree 2.
For these polynomials, we introduce \emph{extended BDDs} (eBDDs) and give eBDD-based algorithms for the following two tasks:
\begin{enumerate}
\item Compute an eBDD representing $\Eval{\psi}$ for every node $\psi$ of $\Convert(\varphi)$.
\item Given an eBDD  for $\Eval{\psi}$ and a partial assignment $\sigma$, compute $\Pev\sigma \Eval{\psi}$.
\end{enumerate}

\parag{Computing eBDDs for {\CPEDG}s: Informal introduction}
Consider a {\CPE} $\varphi$ and its associated {\CPEDG} $\Convert(\varphi)$.
Each node of $\varphi$ induces a chain of nodes in $\Convert(\varphi)$, consisting of degree-reduction nodes $\delta_{x_1}, \ldots, \delta_{x_n}$, followed by the 
node itself (see Figure \ref{fig:chain}). 
Given BDDs $u$ and $v$ for the children of the node in the {\CPE}, we can compute a BDD for the node itself using a well-known BDD algorithm $\bddop(u,v)$ 
parametric in the boolean operation $\circledast$ labelling the node \cite{Bryant86}. Our goal is to transform $\bddop$ into an algorithm that computes eBDDs \emph{for all nodes in the chain}, i.e.\ eBDDs for all the polynomials $p_0, p_1, \ldots, p_n$ of Figure \ref{fig:chain}. 

Roughly speaking, $\bddop(u, v)$ recursively computes BDDs $w_0 = \bddop(u_0, v_0)$ and $w_1 =\bddop(u_1, v_1)$, where $u_b$ and $v_b$ are the $b$-children of $u$ and $v$, and then returns the BDD with $w_0$ and $w_1$ as $0$- and $1$-child, respectively.\footnote{In fact, this is only true when $u$ and $v$ are nodes at the same level and $\bddop(u_0, v_0)\ne\bddop(u_1, v_1)$, but at this point we only want to convey some intuition.}

Most importantly, we modify $\bddop$ to run in breadth-first order. Figure~\ref{fig:BDDop} shows a graphical representation of a run of $\bddopvee(u, v)$, where $u$ and $v$ are the two BDD nodes labelled by $x$. Square nodes represent pending calls to $\bddop$. Initially there is only one square call $\bddopvee(u, v)$ (Figure~\ref{fig:BDDop}, top left). 
$\bddopvee$ calls itself recursively for $u_0, v_0$ and $u_1, v_1$ (Figure~\ref{fig:BDDop}, top right). 
Each of the two calls splits again into two; however, the first three are identical (Figure~\ref{fig:BDDop}, bottom left), and so reduce to two. 
These two calls can now be resolved directly; they return nodes $\True$ and $\False$, respectively. At this point, the children of $\bddop(u, v)$ become $\Bdd{y,\True,\True}=\True$, and $\Bdd{y,\True,\False}$, which exists already as well (Figure \ref{fig:BDDop}, bottom right).

We look at the diagrams of Figure \ref{fig:BDDop} not as a visualisation aid, but as graphs with two kinds of nodes: 
standard BDD nodes, represented as circles, and \emph{product} nodes, represented as squares. We call them \emph{extended BDDs}.
Each node of an extended BDD is assigned a polynomial in the expected way: 
the polynomial $\Eval{u}$ of a standard BDD node $u$ with variable $x$  is $x \cdot \Eval{u_1} + (1-x) \cdot\Eval{u_0}$, 
the polynomial $\Eval{v}$ of a square $\wedge$-node $v$ is $\Eval{v_0} \cdot \Eval{v_1}$, etc. In this way we assign to each eBDD  a polynomial. In particular, we obtain the intermediate  polynomials $p_0, p_1, p_2, p_3$ of the figure, one for each level in the recursion. In the rest of the section we show that these are \emph{precisely} the polynomials $p_0, p_1, \ldots, p_n$ of Figure~\ref{fig:chain}.

Thus, in order to compute eBDDs for all nodes of a {\CPEDG} $\Convert(\varphi)$, it suffices to  compute BDDs for all nodes of the {\CPE} $\varphi$. 
Since we need to do this anyway to solve \Pref{\#\CPE}, the polynomial certification does not incur any overhead.

\begin{figure}[t]
\tikzstyle{bddnode}=[circle,draw,thick,minimum width=4ex,minimum height=3ex]
\tikzstyle{product}=[bddnode,rectangle,draw=niceblue,fill=nicebgblue]
\tikzstyle{edge1}=[thick]
\tikzstyle{edge0}=[edge1,very thick,dashed,nicered]
\tikzstyle{edgep}=[edge1,niceblue]
\tikzstyle{edgel}=[thin,dashed]
\tikzset{-{Latex[length=2mm, width=2mm]},node distance=12mm}
\tikzstyle{diag}=[node distance=12mm]
\tikzstyle{labeln}=[draw=none,inner sep=0mm,text depth=1mm]
\centering
\def\labeldist{24mm}
\def\ldist{0mm}
\ifx\Fullversion\undefined
\def\vdist{42mm}
\else
\def\vdist{45mm}
\fi
\begin{tikzpicture}
\begin{scope}
\node[bddnode] (q2) {$x$};
\node[bddnode, right of =q2] (q3) {$x$};
\node[above = \ldist of q2,labeln] (l2) {$\mathbf{u}$};
\node[above = \ldist of q3,labeln] (l3) {$\mathbf{v}$};
\node[bddnode, below of=q2] (q4) {$y$};
\node[bddnode, below of=q3] (q5) {$y$};
\path (q4) -- (q5) node[bddnode, midway, yshift=-11mm, ellipse] (q6) {$\True$};
\node[product, right=1.5cm of q3] (q7) {$\vee$};
\node[below right = \labeldist and -0.7cm of q2] (lp) {$p_0 = x + y - xy-x^2y-xy^2 +2x^2y^2$};
\node[above = 0.5mm of q7,labeln] (l7) {$\bddopvee(u, v)$};
\draw
(q2) edge[edge1] (q4)
(q2) edge[edge0] (q5)
(q3) edge[edge1] (q5)
(q4) edge[edge0] (q6)
(q5) edge[edge1] (q6)
(q7) edge[edgep] (q3)
(q7) edge[edgep,bend left=25] (q2);
\end{scope}
\begin{scope}[xshift=7cm]
\node[bddnode] (q2) {$x$};
\node[bddnode, right of =q2] (q3) {$x$};
\node[above = \ldist of q2,labeln] (l2) {$\mathbf{u}$};
\node[above = \ldist of q3,labeln] (l3) {$\mathbf{v}$};
\node[bddnode, below of=q2] (q4) {$y$};
\node[bddnode, below of=q3] (q5) {$y$};
\path (q4) -- (q5) node[bddnode, midway, yshift=-11mm, ellipse] (q6) {$\True$};
\node[bddnode, right=1.5cm of q3] (q7) {$x$};
\node[product, below of=q7,xshift=-6mm] (q8) {$\vee$};
\node[product, below of=q7,xshift=6mm] (q9) {$\vee$};
\node[below right = \labeldist and -0.7cm of q2] (lp) {$p_1=  x+y-2xy+xy^2 = \textcolor{niceredbright}{\mathbf{\delta_x p_0}}$};
\node[above = \ldist of q7,labeln] (l7) {$\bddopvee(u, v)$};
\draw 
(q2) edge[edge1] (q4)
(q2) edge[edge0] (q5)
(q3) edge[edge1] (q5)
(q4) edge[edge0] (q6)
(q5) edge[edge1] (q6)
(q7) edge[edge1] (q8)
(q7) edge[edge0] (q9)
(q8) edge[bend left=30, edgep] node{} (q4)
(q8) edge[edgep] node{} (q5)
(q9) edge[bend left=40, edgep] node{} (q5);
\end{scope}
\begin{scope}[yshift=-\vdist]
\node[bddnode] (q2) {$x$};
\node[bddnode, right of =q2] (q3) {$x$};
\node[above = \ldist of q2,labeln] (l2) {$\mathbf{u}$};
\node[above = \ldist of q3,labeln] (l3) {$\mathbf{v}$};
\node[bddnode, below of=q2] (q4) {$y$};
\node[bddnode, below of=q3] (q5) {$y$};
\path (q4) -- (q5) node[bddnode, midway, yshift=-11mm, ellipse] (q6) {$\True$};
\node[bddnode, right=1.5cm of q3] (q7) {$x$};
\node[bddnode, below of=q7,xshift=-6mm] (q8) {$y$};
\node[bddnode, below of=q7,xshift=6mm] (q9) {$y$};
\node[product] (q10) at (q6 -| q8) {$\vee$};
\node[product] (q11) at (q6 -| q9) {$\vee$};
\node[below right = \labeldist and -0.7cm of q2] (lp) {$p_2 = x+y-xy = \textcolor{niceredbright}{\mathbf{\delta_y\delta_x p_0}}$};
\node[above = \ldist of q7,labeln] (l7) {$\bddopvee(u, v)$};
\draw 
(q2) edge[edge1] (q4)
(q2) edge[edge0] (q5)
(q3) edge[edge1] (q5)
(q4) edge[edge0] (q6)
(q5) edge[edge1] (q6)
(q7) edge[edge1] (q8)
(q7) edge[edge0] (q9)
(q8) edge[edge1,bend right = 20] node{} (q10)
(q8) edge[edge0,bend left = 20] node{} (q10)
(q9) edge[edge1] node{} (q10)
(q9) edge[edge0] node{} (q11)
(q10) edge[edgep] node{} (q6);
\end{scope}
\begin{scope}[xshift=7cm,yshift=-\vdist]
\node[bddnode] (q2) {$x$};
\node[bddnode, right of =q2] (q3) {$x$};
\node[above = \ldist of q2,labeln] (l2) {$\mathbf{u}$};
\node[above = \ldist of q3,labeln] (l3) {$\mathbf{v}$};
\node[bddnode, below of=q2] (q4) {$y$};
\node[bddnode, below of=q3] (q5) {$y$};
\path (q4) -- (q5) node[bddnode, midway, yshift=-11mm, ellipse] (q6) {$\True$};
\node[bddnode, right=1.2cm of q3] (q7) {$x$};
\node[below right = \labeldist and -0.7cm of q2] (lp) {$p_3=  x+y-xy= \textcolor{niceredbright}{\mathbf{\delta_y\delta_x p_0}}$};
\node[above = \ldist of q7,labeln] (l7) {$\bddopvee(u, v)$};
\draw 
(q2) edge[edge1] (q4)
(q2) edge[edge0] (q5)
(q3) edge[edge1] (q5)
(q4) edge[edge0] (q6)
(q5) edge[edge1] (q6)
(q7) edge[edge1, bend left] (q6)
(q7) edge[edge0] (q5);
\end{scope}    
\end{tikzpicture}
\caption{Run of $\bddopvee(u,v)$, but with recursive calls evaluated in breadth-first order. All missing edges go to node $\False$.}
\label{fig:BDDop}
\end{figure}
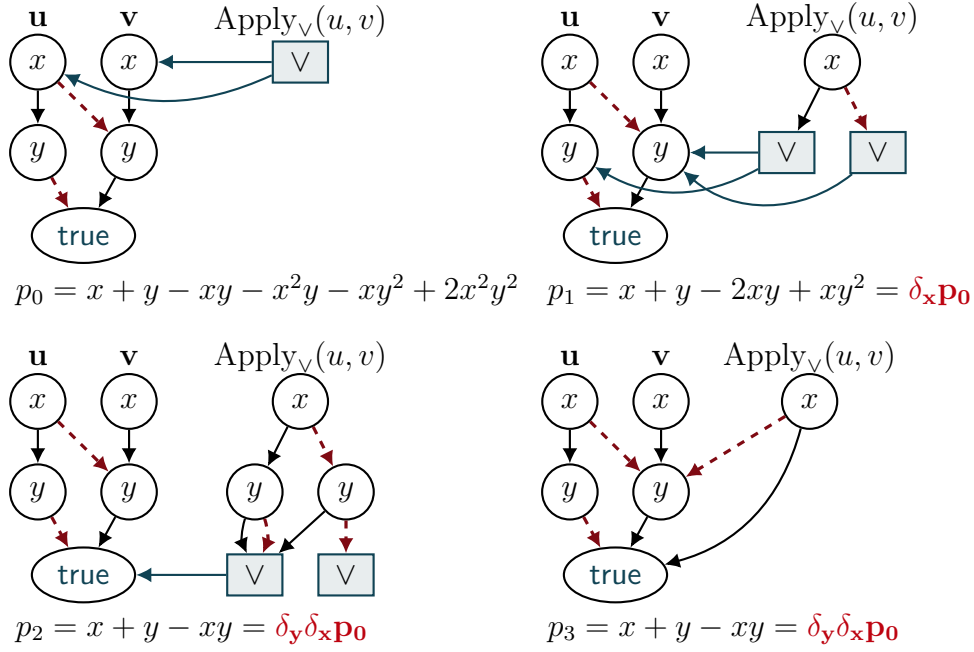

\parag{Extended BDDs}
As for BDDs, we define eBDDs over $X=\{x_1, \ldots, x_n\}$ with the variable order $x_1<x_{2}<...<x_n$.

\begin{definition}
\label{def:eBDD}
Let $\circledast$ be a binary boolean operator. The set of eBDDs (for $\circledast$) is inductively defined as follows:
\begin{itemize}
\item every BDD is also an eBDD of the same level;
\item if $u, v$ are BDDs (not eBDDs!), then $\Bdd{u \circledast v}$ is an eBDD of level $l$ where $l:=\max\{\Level(u), \Level(v) \}$;
we call eBDDs of this form \emph{product nodes};
\item if $u\ne v$ are eBDDs and $i>\Level(u),\Level(v)$, then $\Bdd{x_i, u, v}$ is an eBDD of level $i$;
\item we identify $\Bdd{x_i, u, u}$ and $u$ for an eBDD $u$ and $i>\Level(u)$.
\end{itemize}
The set of \emph{descendants} of an eBDD $w$ is the smallest set $S$ with $w\in S$ and $u,v\in S$ for all $\Bdd{u \circledast v}, \Bdd{x,u,v}\in S$
The \emph{size} of $w$ is its number of descendants.
For $u, v \in \{ \Bdd{\True}, \Bdd{\False} \}$ we identify $\Bdd{u \circledast v}$ with $\Bdd{\True}$ or $\Bdd{\False}$ according to the result of $\circledast$, e.g.\ $\Bdd{\Bdd{\True}\vee \Bdd{\False}}=\Bdd{\True}$, as $\True\vee\False=\True$.
\label{def:earith}
The \emph{arithmetisation} of an eBDD for a boolean operator $\circledast \in \{\wedge, \vee\}$ is defined as for BDDs, with the extensions  $\Eval{\Bdd{u\wedge v}} = \Eval{u}\cdot\Eval{v}$ and $\Eval{\Bdd{u\vee v}} = \Eval{u}+\Eval{v}-\Eval{u}\cdot\Eval{v}$.
\end{definition}

\begin{example}
The diagrams in Figure \ref{fig:BDDop} are eBDDs for $\circledast := \vee$. Nodes of the form $\Bdd{x, u, v}$ and $\Bdd{u \vee v}$ are represented as circles and squares, respectively. Consider the top-left diagram.  Abbreviating $x \oplus y := (x \wedge  \neg y) \vee (\neg x \wedge y)$ we get $\Eval{\bddopvee(u,v)} = \Eval{ (x \oplus y) \wedge (x \wedge y) } =
\Eval{x \oplus y} \cdot \Eval{x \wedge y} = (x (1 -y)+(1-x)\cdot y - xy(1-x)(1-y)) \cdot xy$, which is the polynomial $p_0$ shown in the figure.
\end{example}

\subsubsection{Computing eBDDs for {\CPEDG}s.}

Given a node of a {\CPE} corresponding to a binary operator $\circledast$, 
Prover has to compute polynomials $p_0,\, \delta_{x_1} p_0,\,\ldots,\,\delta_{x_n}...\delta_{x_1} p_0$ corresponding to the nodes of the {\CPEDG} shown on the right. We show that Prover can compute these polynomials by representing them as eBDDs. 
Table~\ref{table:alg} describes an algorithm that gets as input an eBDD $w$ of level $n$, and outputs a sequence $w_0,w_1,...,w_{n+1}$ of eBDDs such that $w_0=w$; $\Eval{w_{i+1}}=\delta_{x_{n-i}}\Eval{w_{i}}$ for every $0 \leq i \leq \Level(w)-1$; and $w_{n+1}$ is a BDD. Interpreted as sequence of eBDDs, Figure~\ref{fig:BDDop} shows a run of this algorithm. 

\emph{Notation.} Given an eBDD $w$ and eBDDs $u, v$ such that $\Level(u) \geq \Level(v)$, we let $w[u/v]$ denote the result of replacing $u$ by $v$ in $w$. For an eBDD $w=\Bdd{x_i,w_0,w_1}$ and $b\in\{0,1\}$ we define $\pev{x_i:=b}w:=w_b$, and for $j>i$ we set $\pev{x_j:=b}w:=w$. (Note that $\Eval{\pev{x_j:=b}w}=\pev{x_j:=b}\Eval{w}$ holds for any $j$ where it is defined.)

\begin{table}[t]
\begin{minipage}[t]{7cm}
\begin{tabbing}
t \= t \= tt \= t \=\kill
$\ComputeEBDD(w)$ \\[0.1cm]
\textbf{Input:} eBDD $w$\\
\textbf{Output:} sequence $w_0,...,w_{n}$ of eBDDs \\[0.1cm]
$w_0:= w$; output $w_0$\\
\textbf{for} $i = 0, \cdots, \Level(w)-1$ \textbf{do}\\
\> $w_{i+1} := w_i$ \\
\> \textbf{for} every node $\Bdd{u \circledast v}$ of $w_i$\\
\> \hspace{5mm}at level $n-i$ \textbf{do}\\
\>\> \textbf{for} $b \in \{0,1\}$ \textbf{do} \\
\>\>\> $u_b := \pev{x_{n-i}:=b} \, u$ \\
\>\>\> $v_b := \pev{x_{n-i}:=b} \, v$ \\
\>\>\> $t_b := \Bdd{u_b \circledast v_b}$ \\
\>\> $w_{i+1} :=  w_{i+1}\left[ \, \Bdd{u \circledast v} /   \Bdd{x_{n-i}, t_0,  t_1}\, \right] $ \\
\> output $w_{i+1}$
\end{tabbing}
\end{minipage}
\hfill
\begin{minipage}[t]{5.5cm}
\newcommand{\Ev}{E_\sigma}
\begin{tabbing}
$\EvaluateEBDD(w,\sigma)=:\Ev(w)$ \\[0.1cm]
\textbf{Input:} \= eBDD $w$; assignment $\sigma \colon X \rightarrow \F$ \\
\textbf{Output:} $\Pev\sigma\Eval{w}$ \\[0.2cm]
\pushtabs
t \= txt \= \kill
\textbf{if} $P(w)$ is defined return $P(w)$ \\
\textbf{if} $w \in\{ \Bdd{\True},\Bdd{\False}\}$ return $\Eval{w}$ \\
\textbf{if} $w = \Bdd{ u \wedge v }$ \\
\> $P(w):= \Ev(u) \cdot \Ev(v)$ \\
\textbf{if} $w = \Bdd{ u \vee v }$ \\
\> $P(w):= \Ev(u)+\Ev(v)-\Ev(u)\Ev(v)$ \\
\textbf{if} $w = \Bdd{x, u, v}$ and $\sigma(x)$ undefined  \\
\> $P(w):= [1-x] \cdot \Ev(u) + [x] \cdot \Ev(v)$ \\
\textbf{if} $w = \Bdd{x, u, v}$ and $\sigma(x)=s\in\F$  \\
\> $P(w):= [1-s] \cdot \Ev(u) + [s] \cdot \Ev(v)$\\
return $P(w)$
\poptabs
\end{tabbing}
\end{minipage}
\smallskip
\caption{On the left: Algorithm computing eBDDs for the sequence $\Eval{w}$, $\delta_{x_{n}} \Eval{w}$, $\delta_{x_{n-1}} \delta_{x_{n}} \Eval{w}$, $\ldots$, $\delta_{x_1} \cdots \delta_{x_{n}} \Eval{w}$ of polynomials. On the right: Recursive algorithm to evaluate the polynomial represented by an eBDD at a given partial assignment. $P(w)$ is a mapping  used to memoize the polynomials returned by recursive calls.}
\label{table:alg}
\ifx\Fullversion\undefined
\vspace*{-0.7cm}
\fi
\end{table}

\begin{restatable}{proposition}{computebdd}\label{prop:computebdd}
Let $\psi_1,\psi_2$ denote \CPEs\ and $u_1,u_2$ BDDs with $\Eval{u_i}=\Eval{\psi_i}$, $i\in\{1,2\}$. Let $w:=\Bdd{u_1\circledast u_2}$ denote an eBDD. Then $\ComputeEBDD(w)$ satisfies $\Eval{w_0}=\Eval{\psi_1\circledast\psi_2}$ and $\Eval{w_{i+1}}=\delta_{x_{n-i}}\Eval{w_{i}}$ for every $0 \leq i \leq n-1$; moreover, $w_n$ is a BDD with $w_n=\bddop(u_1,u_2)$. Finally, the algorithm runs in time $\O(T)$, where $T\in\O(\Abs{u_1}\cdot\Abs{u_2})$ is the time taken by $\bddop(u_1,u_2)$.
\end{restatable}

\parag{Evaluating polynomials represented as eBDDs}
Recall that Prover must evaluate expressions of the form $\Pev{\sigma}\Eval{\psi}$ for some 
{\CPEDG} $\psi$, where $\sigma$ assigns values to all variables of $\psi$ except for possibly one.
We give an algorithm to evaluate arbitrary expressions $\Pev{\sigma}\Eval{w}$, where $w$ is an eBDD, and show that if there is at most
one free variable then the algorithm takes linear time in the size of $\psi$. The algorithm 
is shown on the right of Table \ref{table:alg}. It has the standard structure
of BDD procedures: It recurs on the structure of the eBDD, memoizing the result of recursive calls so that the
algorithm is called at most once with a given input.

\begin{restatable}{proposition}{evaluateebdd}\label{prop:evaluateebdd}
Let $w$ denote an eBDD, $\sigma: X\rightarrow\F$ a partial assignment, and $k$ the number of variables assigned by $\sigma$. 
Then \EvaluateEBDD\ evaluates the polynomial $\Pev\sigma\Eval{w}$ in time $\O\big(\!\operatorname{poly}(2^{n-k})\cdot\Abs{w}\big)$.
\end{restatable}

\subsection{Efficient Certification}
\label{sec:main}

In the \CPCertify\ algorithm, Prover must (a) compute polynomials for all nodes of the {\CPEDG}, and (b) evaluate them on assignments chosen by Verifier. 
In the last section we have seen that \ComputeEBDD\ (for binary operations of the \CPE), combined with standard BDD algorithms (for all other operations), 
yields eBDDs representing all these polynomials---at no additional overhead, compared to a BDD-based implementation. 
This covers part (a). 
Regarding (b), recall that all polynomials computed in (a) have at most one variable. 
Therefore, using \EvaluateEBDD\ we can evaluate a polynomial in linear time in the size of the eBDD representing it. 

The Verifier \CPCertify\ is implemented in a straightforward manner. As the algorithm runs in polynomial size w.r.t.\ the \CPE\ (and not the computed BDDs, which may be exponentially larger), incurring overhead is less of a concern.

\begin{restatable}[Main Result]{theorem}{thmmain}
\label{thm:main}
If \BDDSolver\ solves an instance $\varphi$ of {\Pref{\#\CPE}} with $n$ variables in time $T$, with $T>n\Abs{\varphi}$, then
\begin{enumerate}[label={(\alph*)}]
\item Prover computes eBDDs for all nodes of $\Convert(\varphi)$ in time $\O(T)$,
\item Prover responds to Verifier's challenges in time $\O(nT)$, and
\item Verifier executes \CPCertify\ in time $\O(n^2\Abs{\varphi})$, with failure probability at most $4n\Abs{\varphi}/\Abs{\F}$.
\end{enumerate}
\end{restatable}

As presented above, \EvaluateEBDD\ incurs a factor-of-$n$ overhead, as every node of the \CPEDG\ must be evaluated. 
In our implementation, we use a caching strategy to reduce the complexity of Theorem~\ref{thm:main}(b) to $\O(T)$.

Note that the bounds above assume a uniform cost model. In particular, operations on BDD nodes and finite field arithmetic are assumed to be $\O(1)$. This is a reasonable assumption, as for a constant failure probability $\log\,\Abs{\F}\approx \log n$. Hence the finite field remains small. (It is possible to verify the number of assignments even if it exceeds $\Abs{\F}$, see below.)

\subsection{Implementation concerns}
We list a number of points that are not described in detail in this paper, but need to be considered for an efficient implementation.

\parag{Finite field arithmetic} It is not necessary to use large finite fields. In particular, one can avoid the overhead of arbitrarily sized integers. For our implementation we fix the finite field $\F:=\mathbb{Z}_p$, with $p = 2^{61} -1$ (the largest Mersenne prime to fit in 64 bits).

\parag{Incremental eBDD representation} Algorithm \ComputeEBDD\ computes a sequence of eBDDs. These must not be stored explicitly, otherwise one incurs a space-overhead. Instead, we only store the last eBDD as well as the differences between each subsequent element of the sequence. To evaluate the eBDDs, we then revert to a previous state by applying the differences appropriately.

\parag{Evaluation order} It simplifies the implementation if Prover only needs to evaluate nodes of the \CPEDG\ in some (fixed) topological order. \CPCertify\ can easily be adapted to guarantee this, by picking the next node appropriately in each iteration, and by evaluating only one child of a binary operator $\psi_1\circledast\psi_2$. The value of the other child can then be derived by solving a linear equation. 

\parag{Efficient evaluation} As stated in Theorem~\ref{thm:main}, using \EvaluateEBDD\ Prover needs $\Omega(nT)$ time to respond to Verifier's challenges. In our implementation we instead use a caching strategy that reduces this time to $\O(T)$. Essentially, we exploit the special structure of $\Convert(\varphi)$: Verifier sends a sequence of challenges
\[\Pev{\sigma_0}\delta_{x_1}...\delta_{x_n}w,\quad\Pev{\sigma_1}\delta_{x_2}...\delta_{x_n}w,\quad ...,\quad \Pev{\sigma_n}w\]
where assignments $\sigma_i$ and $\sigma_{i+1}$ differ only in variables $x_i$ and $x_{i+1}$. The corresponding eBDDs likewise change only at levels $i$ and $i+1$. We cache the linear coefficients of eBDD nodes that contribute to the arithmetisation of the root top-down, and the arithmetised values of nodes bottom up. As a result, only levels $i,i+1$ need to be updated.

\parag{Large numbers of assignments} If the number of satisfying assignments of a \CPE\ exceeds $\Abs{\F}$, Verifier would not be able to verify the count accurately. Instead of choosing $\Abs{\F}\ge2^n$, which incurs a significant overhead, Verifier can query the precise number of assignments, and then choose $\Abs{\F}$ randomly. This introduces another possibility of failure, but (roughly speaking) it suffices to double $\log\,\Abs{\F}$ for the additional failure probability to match the existing one. Our implementation does not currently support this technique.

\section{Evaluation}\label{sec:evaluation}

We have implemented an eBDD library, \blic{} (BDD Library with Interactive Certification)~\footnote{\url{https://gitlab.lrz.de/i7/blic}}, 
that is a stand-in replacement for BDDs but additionally performs the role of Prover in the \textsc{CPCertify} protocol.
We have also implemented a client that executes the protocol as Verifier.
The eBDD library is about 900 lines of C++ code and the \textsc{CPCertify} protocol is about 400 lines.
We have built a prototype certifying QBF solver in \blic, totalling about 2600 lines of code.
We aim to answer the following questions in our evaluation:

\begin{enumerate}[align=left]
\item[\textbf{RQ1.}] Is a QBF solver with \textsc{CPCertify}-based certification competitive?
If so, how high is the overhead of implementing \textsc{CPCertify} on top of the BDD operations?
\item[\textbf{RQ2.}] What is the amount of communication for Prover and Verifier in executing the \textsc{CPCertify}
protocol, what is the time requirement for Verifier, and how do these numbers compare to proof 
sizes and proof checking times for certificates based on resolution and other proof systems? 
\end{enumerate}

\parag{RQ1: Performance of \blic}
We compare \blic{} with  \caqe{}, \depqbf{}, and \pgbddq{},
three state-of-the-art QBF solvers. 
\caqe{} \cite{CAQE,CAQETool} does not provide any certificates in its most recent version.
\depqbf{} \cite{LonsingE17,depqbf} is a certifying QBF solver.
\pgbddq{} \cite{BryantH21qbf,PGBDDQ} is an independent implementation of a BDD-based QBF solver.
Both \depqbf{} and \pgbddq{} provide specialised checkers for their certificates, though \pgbddq{} can also proofs in standard QRAT format.
Note that \pgbddq{} is written in Python and generates proofs in an ASCII-based format, incurring overhead compared to the other tools.

    \begin{figure}[t]\centering
\ifx\Fullversion\undefined
	\def\w{59.6mm}
	\def\h{0.5mm}
	\def\hh{1mm}
\else
	\def\w{73mm}
	\def\h{3mm}
	\def\hh{5mm}
\fi
{\footnotesize
        \begin{tabular}{cc}
        \includegraphics[width=\w]{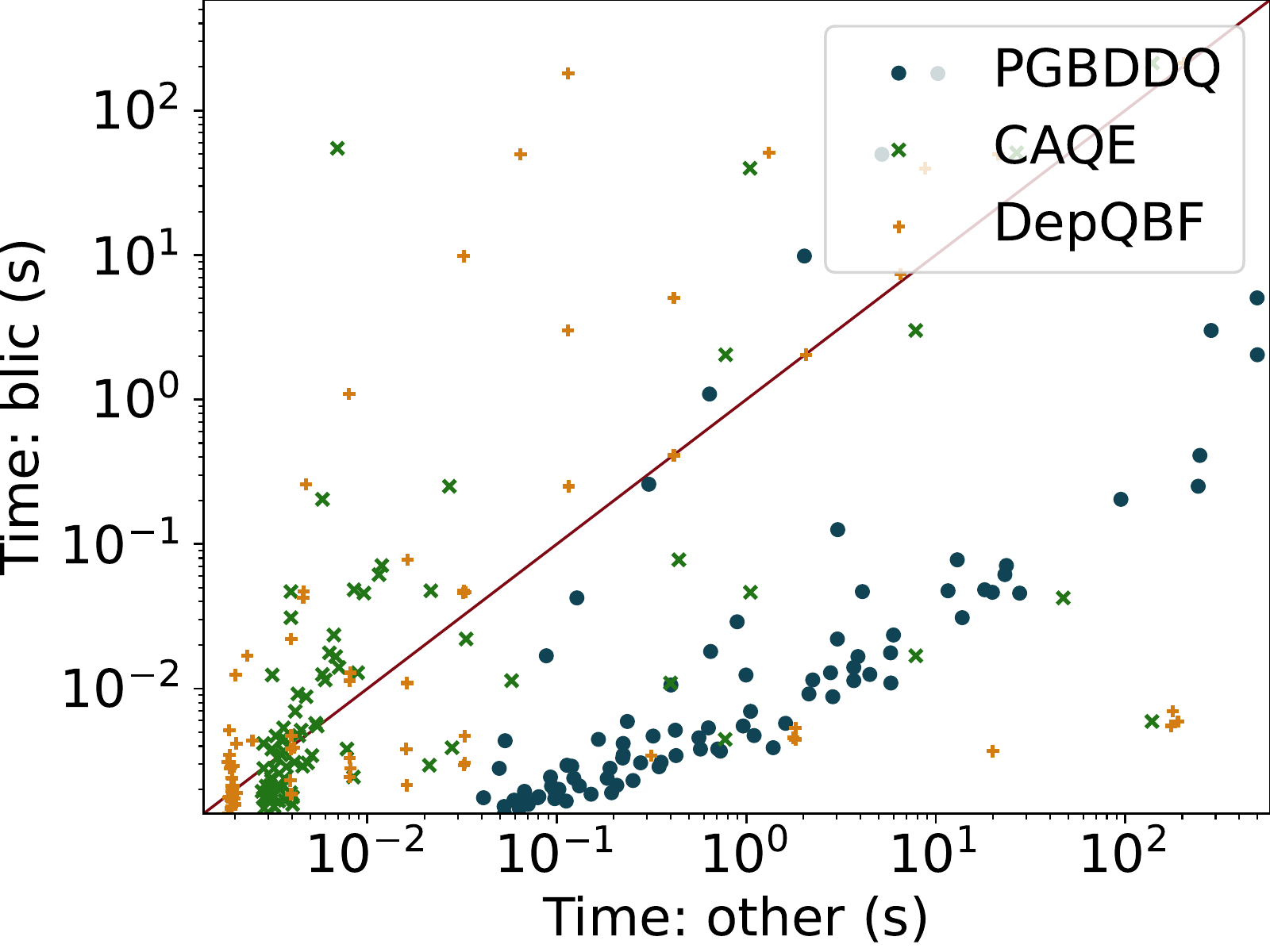}\hspace{\h}
&
        \includegraphics[width=\w]{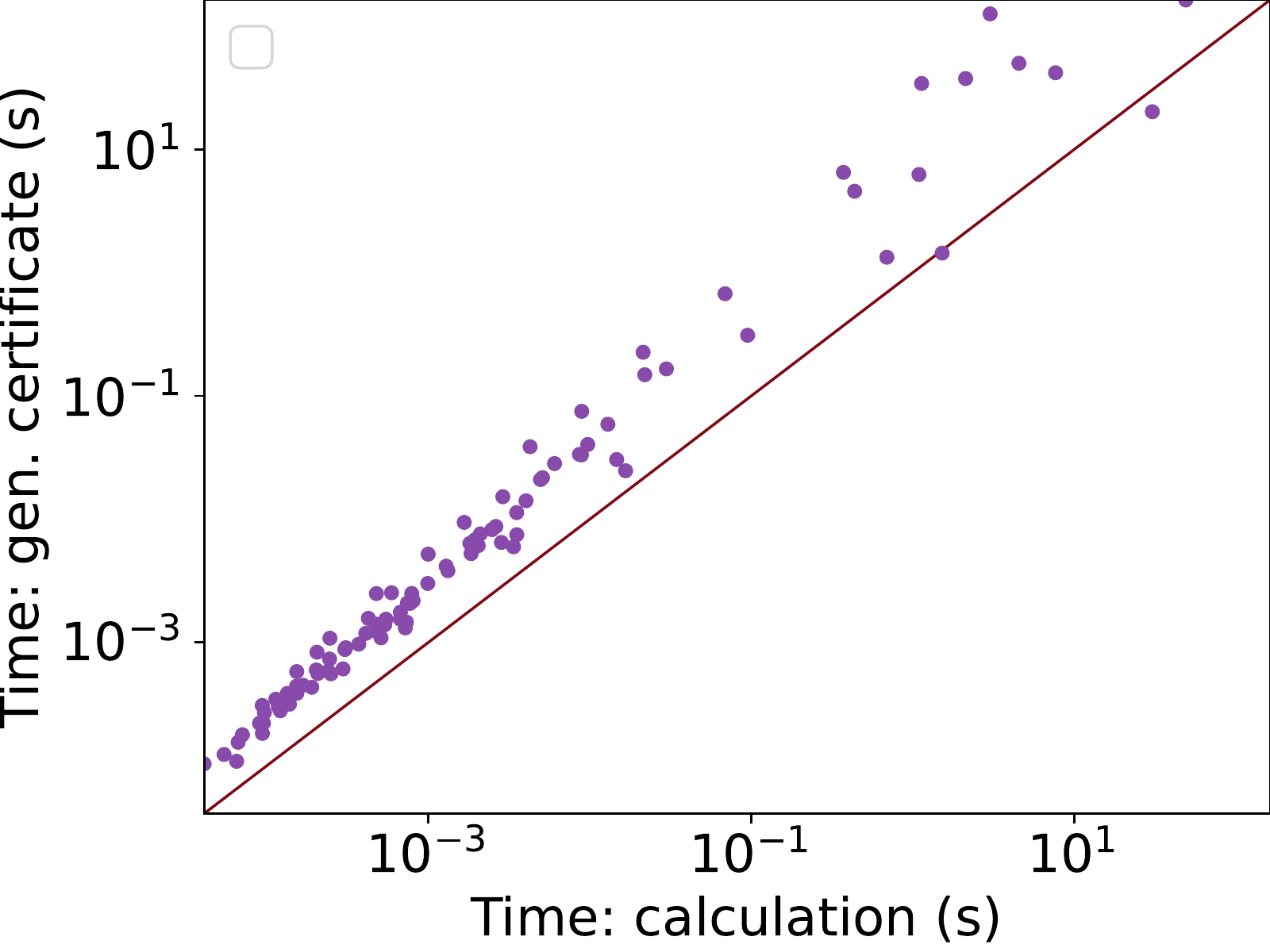} \\
	(a) & (b) \\[\hh]
        \includegraphics[width=\w]{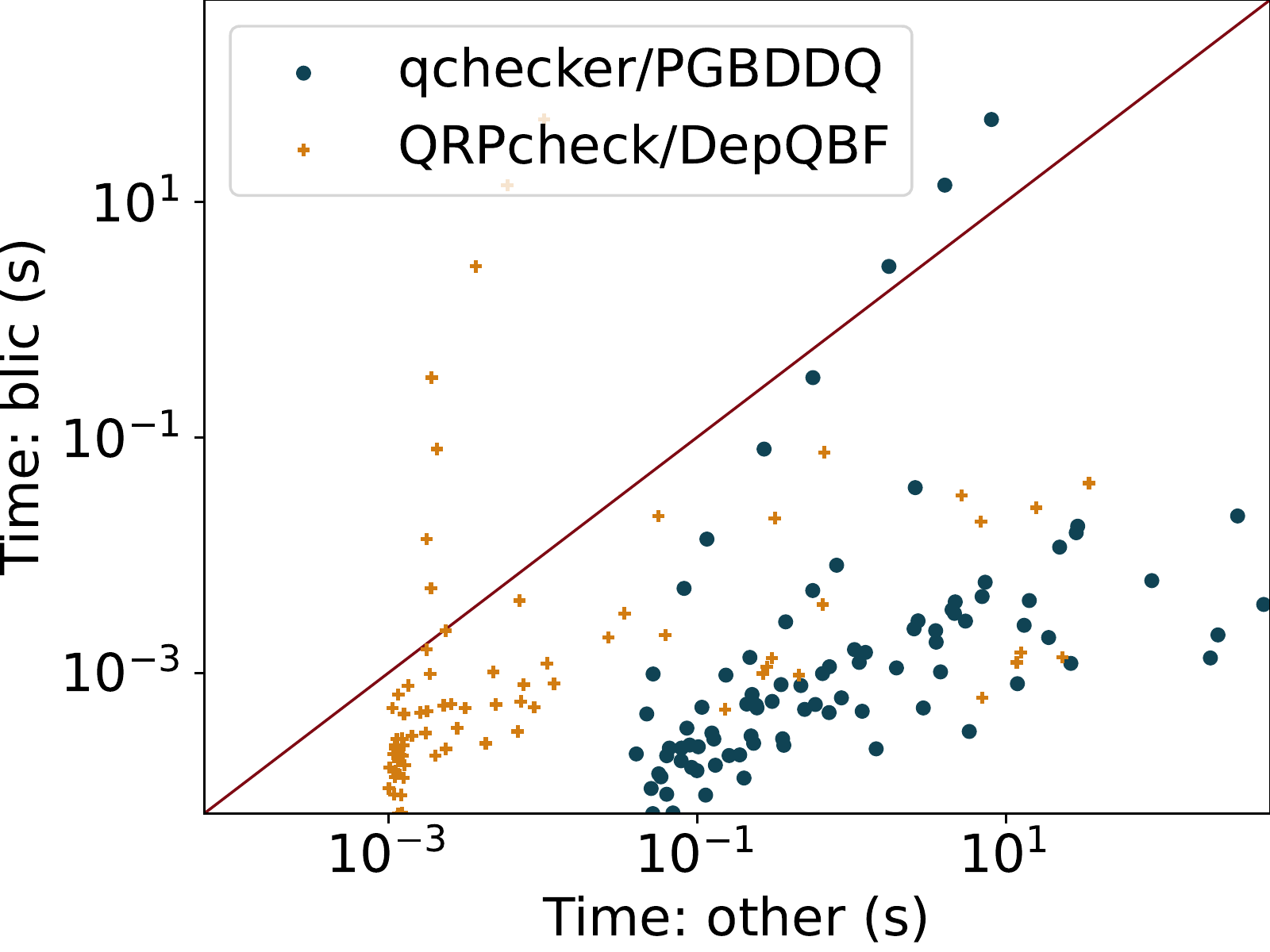} \hspace{\h}
&
        \includegraphics[width=\w]{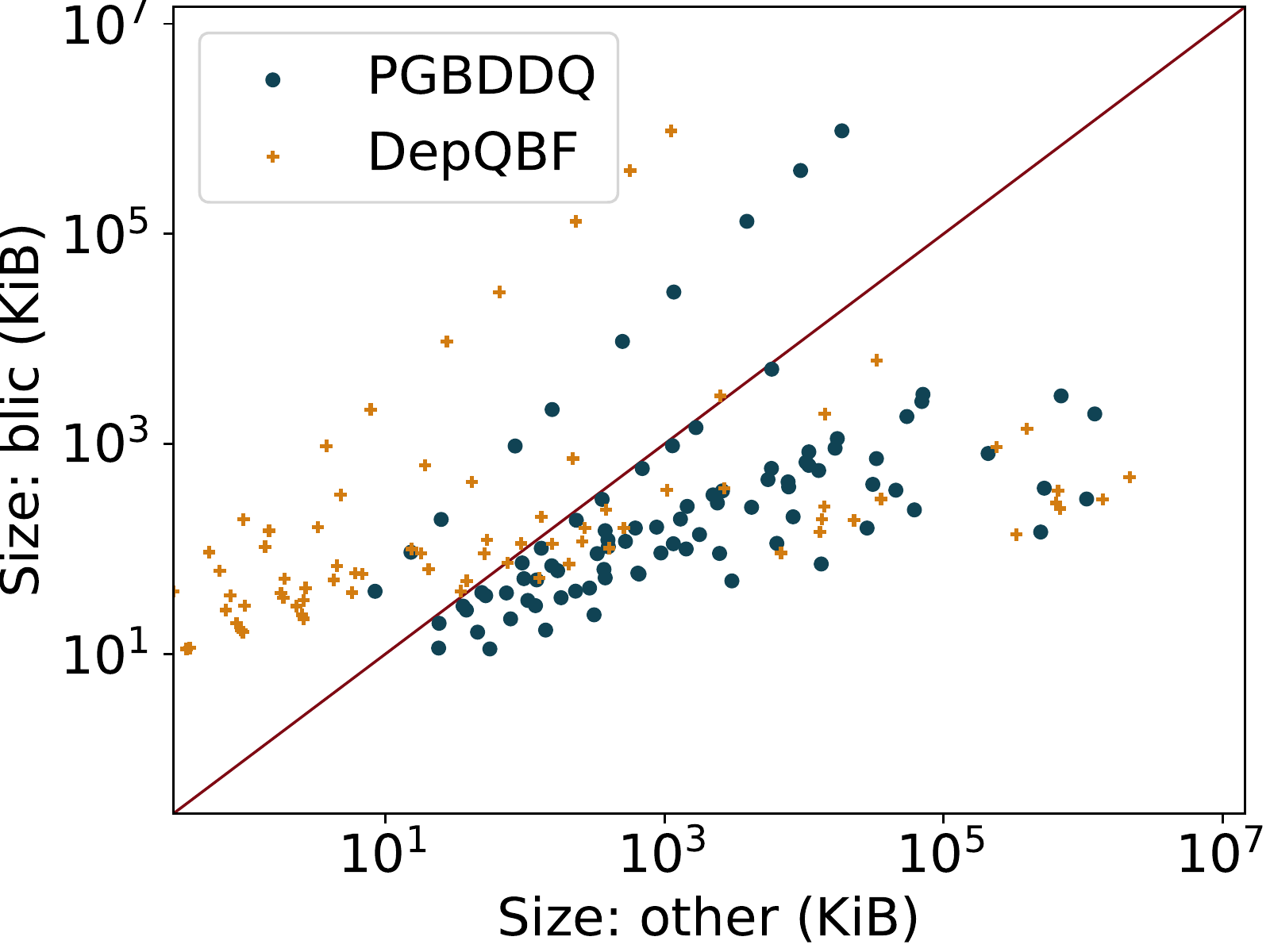} \\
	(c) & (d) \\
        \end{tabular}
}
        \caption{(a) Time taken on instances (dashed lines are $y=100x$ and $y=0.01x$),
                 (b) Cost of generating a certificate over computing the solution, 
        		 (c) Time to verify the certificate,
        		 (d) Size of certificates
    	}	
\label{fig:plots}
        \end{figure}

We take 172 QBF instances (all unsatisfiable) from the \emph{Crafted Instances} track of the QBF Evaluation 2022.\footnote{\caqe{} and \depqbf{} were the winner and runner-up in this category. The configuration we used differs from the competition, as described in \appref{app:setup}.
}
The \emph{Prenex CNF} track of the QBF competition is not evaluated here. It features instances with a large number of variables. BDD-based solvers perform poorly under these circumstances without additional optimisations. Our overall goal is not to propose a new approach for solving QBF, but rather to certify a BDD-based approach, so we wanted to focus on cases where the existing BDD-based approaches are practical.

We ran each benchmark with a 10 minute timeout; all tools other than \caqe{} were run with certificate production.
All times were obtained on a machine with an Intel Xeon E7-8857 CPU and 1.58 TiB RAM\footnote{\blic{} uses at most 60 GiB on the shown benchmarks, 5 GiB when excluding timeouts.} running Linux. See \appref{app:setup} for a detailed description.
\blic{} solved 96 out of 172 benchmarks, \caqe{} solved 98, \depqbf{} solved 87, and \pgbddq{} solved 91.
Figure~\ref{fig:plots}(a) shows the run times of \blic{} compared to the other tools.
The plot indicates that \blic{} is competitive on these instances, with a few cases, mostly from the Lonsing family of benchmarks, 
where \blic{} is slower than \depqbf{} by an order of magnitude. 
Figure~\ref{fig:plots}(b) shows the overhead of certification: for each benchmark (that finishes within a 10min timeout),
we plot the ratio of the time to compute the answer to the time it takes to run Prover in \textsc{CPCertify}. 
The dotted regression line shows \CPCertify{} has a 2.8$\times$ overhead over computing BDDs. 
For this set of examples, the error probability never exceeds $10^{-8.9}$ ($10^{-11.6}$ when Lonsing examples are excluded);
running the verifier $k$ times reduces it to $10^{-8.9k}$.

\begin{table}[t]
\begin{center}
\setlength{\tabcolsep}{1.5mm}
{\small
\centering
    \begin{tabular}{llrrrrrr}
\multicolumn{2}{c}{Instance} & \multicolumn{2}{c}{Solve time (s)} & \multicolumn{2}{c}{Certificate (MiB)} & \multicolumn{2}{c}{Verifier time (s)} \\
\multicolumn{1}{c}{$n$}& \multicolumn{1}{c}{result}& \multicolumn{1}{c}{\blic{}} & \multicolumn{1}{c}{\pgbddq{}} & \multicolumn{1}{c}{\blic{}} & \multicolumn{1}{c}{\pgbddq{}} & \multicolumn{1}{c}{\blic{}} & \multicolumn{1}{c}{\qchecker{}} \\ \cmidrule{1-2} \cmidrule(lr){3-4} \cmidrule(lr){5-6} \cmidrule(lr){7-8}
10 & sat & 0.03 & 3.67 & 1.20 & 8.48 & 0.01 & 3.80 \\
10 & unsat & 0.03 & 3.66 & 1.20 & 8.45 & 0.01 & 3.83 \\
15 & sat & 0.13 & 18.07 & 4.12 & 44.25 & 0.02 & 18.45 \\
15 & unsat & 0.13 & 18.14 & 4.11 & 44.20 & 0.02 & 18.55 \\
20 & sat & 0.54 & 82.92 & 11.59 & 198.54 & 0.07 & 80.28 \\
20 & unsat & 0.53 & 83.02 & 11.64 & 198.76 & 0.06 & 79.05 \\
25 & sat & 1.56 & 261.16 & 23.94 & 566.95 & 0.14 & 238.99 \\
25 & unsat & 1.55 & 261.25 & 23.86 & 565.36 & 0.15 & 237.94 \\
40 & sat & 25.22 & 4863.71 & 132.43 & 7464.96 & 0.95 & 5141.08 \\
40 & unsat & 25.25 & 4827.06 & 132.67 & 7467.84 & 0.99 & 5463.54 \\
    \end{tabular}
}
\end{center}
\caption{Comparison of certificate generation, bytes exchanged between prover and verifier, and time taken to verify the certificate on a set of QBF
benchmarks from \cite{BryantH21qbf}. ``Solve time'' is time taken to solve the instance and to generate a certificate (seconds), ``Certificate'' is the 
size of proof encoding for \pgbddq{}, and bytes exchanged by \CPCertify{} for \blic{}, and ``Verifier time'' is time to verify the
certificate (Verifier's run time for \blic{} and time taken by \qchecker).}
\label{table:pgbddqvsblic}
\end{table}

\parag{RQ2: Communication Cost of Certification and Verifier Time}
We explore \textbf{RQ2} by comparing the number of bytes exchanged between Prover and Verifier and the time needed for 
Verifier to execute \CPCertify{} 
with the number of bytes in an QBF proof and the time required to verify the proof produced by \depqbf{} and \pgbddq{}, for which we use \qrpcheck{}~\cite{NiemetzPLSB12,QRPCHECK} and \qchecker~\cite{BryantH21qbf,PGBDDQ}, respectively. Note that the latter is written in Python.

We show that the overhead of certification is low. 
Figure~\ref{fig:plots}(c) shows the run time of Verifier---this is generally negligible for \blic{}, except for the Lonsing and KBKF families, which have a large number of variables, but very small BDDs.
Figure~\ref{fig:plots}(d) shows the total number of bytes exchanged between Prover and Verifier in \blic{} against 
the size of the proofs generated by \pgbddq{} and \depqbf{}.
For large instances, the number of bytes exchanged in \blic{} is significantly smaller than the 
size of the proofs. The exception are again the Lonsing and KBKNF families of instances. For both plots, the dotted line results from a log-linear regression.

In addition to the Crafted Instances, we compare
against \pgbddq{} on a challenging family of benchmarks used in the \pgbddq{} 
paper (matching the parameters of \cite[Table 3]{BryantH21qbf}); these are QBF encodings of a linear domino placing game.\footnote{\depqbf{} only solved 1 of 10 instances within 120min, and is thus not compared.}
Our results are summarised in Table~\ref{table:pgbddqvsblic}.
The upper bound on Verifier error is $10^{-9.22}$.
We show that \blic\ outperforms \pgbddq{} both in overall cost of computing the answer and the certificates as well as in the number of bytes communicated and the time used by Verifier.

Our results indicate that giving up absolute certainty through interactive protocols can lead to an order of magnitude smaller communication cost and several orders of magnitude
smaller checking costs for the verifier.

\section{Conclusion}

We have presented a solver that combines BDDs with an interactive protocol. 
\blic{} can be seen as a self-certifying BDD library able to certify the correctness of arbitrary sequences of BDD operations.
In order to trust the result, a user must only trust the verifier (a straightforward program that poses challenges to the prover).
We have shown that \blic{} (including certification time) is competitive with other solvers,
and Verifier's time and error probabilities are negligible.

Our results show that $\IP=\PSPACE$ can become an important result not only in theory but also in the practice of automatic verification. 
From this perspective, our paper is a first step towards practical certification based on interactive protocols. 
While we have focused on BDDs, we can ask the more general question: which practical automated reasoning algorithms can be made efficiently certifying? 
For example, whether there is an interactive protocol and an efficient certifying version of modern SAT solving algorithms is an interesting open challenge.
 
\bibliographystyle{plain}
\bibliography{references}

\begin{thebibliography}{10}

\bibitem{AroraBarak}
S.~Arora and B.~Barak.
\newblock {\em Computational Complexity: A Modern Approach}.
\newblock Cambridge University Press, 2006.

\bibitem{Babai}
L{\'{a}}szl{\'{o}} Babai.
\newblock Trading group theory for randomness.
\newblock In Robert Sedgewick, editor, {\em Proceedings of the 17th Annual
  {ACM} Symposium on Theory of Computing, May 6-8, 1985, Providence, Rhode
  Island, {USA}}, pages 421--429. {ACM}, 1985.

\bibitem{BalabanovWJ14}
Valeriy Balabanov, Magdalena Widl, and Jie{-}Hong~R. Jiang.
\newblock {QBF} resolution systems and their proof complexities.
\newblock In Carsten Sinz and Uwe Egly, editors, {\em Theory and Applications
  of Satisfiability Testing - {SAT} 2014 - 17th International Conference, Held
  as Part of the Vienna Summer of Logic, {VSL} 2014, Vienna, Austria, July
  14-17, 2014. Proceedings}, volume 8561 of {\em Lecture Notes in Computer
  Science}, pages 154--169. Springer, 2014.

\bibitem{BarbosaRKLNNOPV22}
Haniel Barbosa, Andrew Reynolds, Gereon Kremer, Hanna Lachnitt, Aina Niemetz,
  Andres N{\"{o}}tzli, Alex Ozdemir, Mathias Preiner, Arjun Viswanathan, Scott
  Viteri, Yoni Zohar, Cesare Tinelli, and Clark~W. Barrett.
\newblock Flexible proof production in an industrial-strength {SMT} solver.
\newblock In Jasmin Blanchette, Laura Kov{\'{a}}cs, and Dirk Pattinson,
  editors, {\em Automated Reasoning - 11th International Joint Conference,
  {IJCAR} 2022, Haifa, Israel, August 8-10, 2022, Proceedings}, volume 13385 of
  {\em Lecture Notes in Computer Science}, pages 15--35. Springer, 2022.

\bibitem{Ben-OrGGHKMR88}
Michael Ben{-}Or, Oded Goldreich, Shafi Goldwasser, Johan H{\aa}stad, Joe
  Kilian, Silvio Micali, and Phillip Rogaway.
\newblock Everything provable is provable in zero-knowledge.
\newblock In Shafi Goldwasser, editor, {\em Advances in Cryptology - {CRYPTO}
  '88, 8th Annual International Cryptology Conference, Santa Barbara,
  California, USA, August 21-25, 1988, Proceedings}, volume 403 of {\em Lecture
  Notes in Computer Science}, pages 37--56. Springer, 1988.

\bibitem{BryantBH22}
Randal~E. Bryant, Armin Biere, and Marijn J.~H. Heule.
\newblock Clausal proofs for pseudo-boolean reasoning.
\newblock In Dana Fisman and Grigore Rosu, editors, {\em Tools and Algorithms
  for the Construction and Analysis of Systems - 28th International Conference,
  {TACAS} 2022, Held as Part of the European Joint Conferences on Theory and
  Practice of Software, {ETAPS} 2022, Munich, Germany, April 2-7, 2022,
  Proceedings, Part {I}}, volume 13243 of {\em Lecture Notes in Computer
  Science}, pages 443--461. Springer, 2022.

\bibitem{BryantH21qbf}
Randal~E. Bryant and Marijn J.~H. Heule.
\newblock Dual proof generation for quantified boolean formulas with a
  bdd-based solver.
\newblock In Andr{\'{e}} Platzer and Geoff Sutcliffe, editors, {\em Automated
  Deduction - {CADE} 28 - 28th International Conference on Automated Deduction,
  Virtual Event, July 12-15, 2021, Proceedings}, volume 12699 of {\em Lecture
  Notes in Computer Science}, pages 433--449. Springer, 2021.

\bibitem{BryantH21sat}
Randal~E. Bryant and Marijn J.~H. Heule.
\newblock Generating extended resolution proofs with a bdd-based {SAT} solver.
\newblock In Jan~Friso Groote and Kim~Guldstrand Larsen, editors, {\em Tools
  and Algorithms for the Construction and Analysis of Systems - 27th
  International Conference, {TACAS} 2021, Held as Part of the European Joint
  Conferences on Theory and Practice of Software, {ETAPS} 2021, Luxembourg
  City, Luxembourg, March 27 - April 1, 2021, Proceedings, Part {I}}, volume
  12651 of {\em Lecture Notes in Computer Science}, pages 76--93. Springer,
  2021.

\bibitem{Bryant86}
R.E. Bryant.
\newblock Graph-based algorithms for {B}oolean function manipulation.
\newblock {\em IEEE Transactions on Computers}, C-35(8):677--691, 1986.

\bibitem{CAQETool}
CAQE.
\newblock \url{https://github.com/ltentrup/caqe}, 2023.
\newblock [Online; accessed 03-February-2023].

\bibitem{depqbf}
depqbf.
\newblock \url{https://github.com/lonsing/depqbf}, 2017.
\newblock [Online; accessed 03-February-2023].

\bibitem{GoldwasserMicaliRackoff}
Shafi Goldwasser, Silvio Micali, and Charles Rackoff.
\newblock The knowledge complexity of interactive proof-systems (extended
  abstract).
\newblock In Robert Sedgewick, editor, {\em Proceedings of the 17th Annual
  {ACM} Symposium on Theory of Computing, May 6-8, 1985, Providence, Rhode
  Island, {USA}}, pages 291--304. {ACM}, 1985.

\bibitem{CAV02BLAST}
T.A. Henzinger, R.~Jhala, R.~Majumdar, G.C. Necula, G.~Sutre, and W.~Weimer.
\newblock Temporal-safety proofs for systems code.
\newblock In {\em CAV 02: Computer-Aided Verification}, Lecture Notes in
  Computer Science 2404, pages 526--538. Springer-Verlag, 2002.

\bibitem{Heule17}
Marijn Heule.
\newblock Everything's bigger in texas: "the largest math proof ever".
\newblock In Christoph Benzm{\"{u}}ller, Christine~L. Lisetti, and Martin
  Theobald, editors, {\em {GCAI} 2017, 3rd Global Conference on Artificial
  Intelligence, Miami, FL, USA, 18-22 October 2017}, volume~50 of {\em EPiC
  Series in Computing}, pages 1--5. EasyChair, 2017.

\bibitem{Heule21}
Marijn J.~H. Heule.
\newblock Proofs of unsatisfiability.
\newblock In Armin Biere, Marijn Heule, Hans van Maaren, and Toby Walsh,
  editors, {\em Handbook of Satisfiability - Second Edition}, volume 336 of
  {\em Frontiers in Artificial Intelligence and Applications}, pages 635--668.
  {IOS} Press, 2021.

\bibitem{JussilaSB06}
Toni Jussila, Carsten Sinz, and Armin Biere.
\newblock Extended resolution proofs for symbolic {SAT} solving with
  quantification.
\newblock In Armin Biere and Carla~P. Gomes, editors, {\em Theory and
  Applications of Satisfiability Testing - {SAT} 2006, 9th International
  Conference, Seattle, WA, USA, August 12-15, 2006, Proceedings}, volume 4121
  of {\em Lecture Notes in Computer Science}, pages 54--60. Springer, 2006.

\bibitem{KatzBTRH16}
Guy Katz, Clark~W. Barrett, Cesare Tinelli, Andrew Reynolds, and Liana
  Hadarean.
\newblock Lazy proofs for dpll(t)-based {SMT} solvers.
\newblock In Ruzica Piskac and Muralidhar Talupur, editors, {\em 2016 Formal
  Methods in Computer-Aided Design, {FMCAD} 2016, Mountain View, CA, USA,
  October 3-6, 2016}, pages 93--100. {IEEE}, 2016.

\bibitem{LonsingE17}
Florian Lonsing and Uwe Egly.
\newblock Depqbf 6.0: {A} search-based {QBF} solver beyond traditional {QCDCL}.
\newblock In Leonardo de~Moura, editor, {\em Automated Deduction - {CADE} 26 -
  26th International Conference on Automated Deduction, Gothenburg, Sweden,
  August 6-11, 2017, Proceedings}, volume 10395 of {\em Lecture Notes in
  Computer Science}, pages 371--384. Springer, 2017.

\bibitem{LundFKN92}
Carsten Lund, Lance Fortnow, Howard~J. Karloff, and Noam Nisan.
\newblock Algebraic methods for interactive proof systems.
\newblock {\em J. {ACM}}, 39(4):859--868, 1992.

\bibitem{LuoAHPTW22}
Ning Luo, Timos Antonopoulos, William~R. Harris, Ruzica Piskac, Eran Tromer,
  and Xiao Wang.
\newblock Proving {UNSAT} in zero knowledge.
\newblock In Heng Yin, Angelos Stavrou, Cas Cremers, and Elaine Shi, editors,
  {\em Proceedings of the 2022 {ACM} {SIGSAC} Conference on Computer and
  Communications Security, {CCS} 2022, Los Angeles, CA, USA, November 7-11,
  2022}, pages 2203--2217. {ACM}, 2022.

\bibitem{Namjoshi01}
K.~Namjoshi.
\newblock Certifying model checkers.
\newblock In {\em CAV 01: Computer Aided Verification}, Lecture Notes in
  Computer Science 2102, pages 2--13. Springer-Verlag, 2001.

\bibitem{NeculaPCC-POPL97}
G.C. Necula.
\newblock Proof-carrying code.
\newblock In {\em Principles of Programming Languages}, pages 106--119. ACM
  Press, 1997.

\bibitem{NiemetzPLSB12}
Aina Niemetz, Mathias Preiner, Florian Lonsing, Martina Seidl, and Armin Biere.
\newblock Resolution-based certificate extraction for {QBF} - (tool
  presentation).
\newblock In Alessandro Cimatti and Roberto Sebastiani, editors, {\em Theory
  and Applications of Satisfiability Testing - {SAT} 2012 - 15th International
  Conference, Trento, Italy, June 17-20, 2012. Proceedings}, volume 7317 of
  {\em Lecture Notes in Computer Science}, pages 430--435. Springer, 2012.

\bibitem{PGBDDQ}
PGBDDQ.
\newblock \url{https://github.com/rebryant/pgbdd}, 2023.
\newblock [Online; accessed 03-February-2023].

\bibitem{QRPCHECK}
QRPcheck.
\newblock \url{http://fmv.jku.at/qrpcheck/}, 2023.
\newblock [Online; accessed 03-February-2023].

\bibitem{Shamir92}
Adi Shamir.
\newblock {IP} = {PSPACE}.
\newblock {\em J. {ACM}}, 39(4):869--877, 1992.

\bibitem{SinzB06}
Carsten Sinz and Armin Biere.
\newblock Extended resolution proofs for conjoining bdds.
\newblock In Dima Grigoriev, John Harrison, and Edward~A. Hirsch, editors, {\em
  Computer Science - Theory and Applications, First International Symposium on
  Computer Science in Russia, {CSR} 2006, St. Petersburg, Russia, June 8-12,
  2006, Proceedings}, volume 3967 of {\em Lecture Notes in Computer Science},
  pages 600--611. Springer, 2006.

\bibitem{CAQE}
Leander Tentrup and Markus~N. Rabe.
\newblock Clausal abstraction for {DQBF}.
\newblock In Mikol{\'{a}}s Janota and In{\^{e}}s Lynce, editors, {\em Theory
  and Applications of Satisfiability Testing - {SAT} 2019 - 22nd International
  Conference, {SAT} 2019, Lisbon, Portugal, July 9-12, 2019, Proceedings},
  volume 11628 of {\em Lecture Notes in Computer Science}, pages 388--405.
  Springer, 2019.

\end{thebibliography}

\appendix

\section{Proof of Lemma~\ref{lem:conv}}

Before we start with the proof, we show a technical property, namely that binary equivalence is preserved by addition and multiplication.
\begin{lemma}\label{lem:equivfacts}
Let $p_1,p_2,q_1,q_2$ denote polynomials with $p_i\equiv_b q_i$ for $i\in\{1,2\}$. Then $p_1+p_2\equiv_b q_1+q_2$, $p_1\cdot p_2\equiv_b q_1\cdot q_2$, and $\pev{x:=r}p_1\equiv_b\pev{x:=r}q_1$ for $r\in\{0,1\}$.
\end{lemma}
\begin{proof}
Let $\sigma$ denote an arbitrary binary assignment. Then 
\begin{equation*}
\begin{array}{llll}
\Pev{\sigma}(p_1+p_2)
&=\Pev{\sigma}p_1+\Pev{\sigma}p_2
&=\Pev{\sigma}q_1+\Pev{\sigma}q_2
&=\Pev{\sigma}(q_1+q_2)\\[1mm]
\Pev{\sigma}(p_1\cdot p_2)
&=\Pev{\sigma}p_1\cdot \Pev{\sigma}p_2
&=\Pev{\sigma}q_1\cdot \Pev{\sigma}q_2
&=\Pev{\sigma}(q_1\cdot q_2)\\
\Pev{\sigma}\pev{x:=r}p_1
&=\pev{x:=r}\Pev{\sigma}p_1
&=\pev{x:=r}\Pev{\sigma}q_1
&=\Pev{\sigma}\pev{x:=r}q_1
\end{array}
\end{equation*}
\Qed
\end{proof}

\lemconv*
\begin{proof}
\newcommand{\starequiv}{\mathbin{\raisebox{-1pt}[0pt]{$\stackrel{(*)}{\equiv}_b$}}}
\proofparag{Part (a)} We proceed by structural induction on $\varphi$, to show $\Eval{\Convert(\varphi)}\equiv_b\Eval{\varphi}$.
\begin{itemize}
\item The base case $\varphi\in\{\True,\False\}$ is trivial.
\item Let $\varphi=\neg\psi$. By hypothesis, $\Eval{\Convert(\psi)}\equiv_b\Eval{\psi}$. Now we have \[\Eval{\Convert(\neg\psi)}=\Eval{\neg\Convert(\psi)}=1-\Eval{\Convert(\psi)}
\starequiv 1-\Eval{\psi}=\Eval{\varphi}
\]
At $(*)$ we use both the induction hypothesis and Lemma~\ref{lem:equivfacts}. We use $(*)$ for the other cases as well, with the same meaning.
\item Let $\varphi=\pev{x:=a}\psi$.
\[\Eval{\Convert(\pev{x:=a}\psi)}
=\pev{x:=\Eval{a}}\Eval{\Convert(\psi)}
\starequiv\pev{x:=a}\Eval{\psi}
=\Eval{\varphi}
\]
\item Let $\varphi=\psi_1\wedge\psi_2$. We use $\delta_xp\equiv_bp$ for any polynomial $p$ and variable $x$.
\begin{gather*}
\Eval{\Convert(\psi_1\wedge\psi_2)}
=\Eval{\delta_{x_1}...\delta_{x_k}(\Convert(\psi_1)\wedge\Convert(\psi_2))}\\\quad
\equiv_b\Eval{\Convert(\psi_1)\wedge\Convert(\psi_2)}
=\Eval{\Convert(\psi_1)}\cdot\Eval{\Convert(\psi_2)}
\starequiv\Eval{\psi_1}\cdot\Eval{\psi_2}\\\quad
=\Eval{\psi_1\wedge\psi_2}=\Eval{\varphi}
\end{gather*}
\item For $\varphi=\psi_1\vee\psi_2$ the argument is analogous.
\end{itemize}

It remains to show that $\Eval{\Convert(\varphi)}$ is multilinear. Again, we do a structural induction. The base case $\varphi\in\{\True,\False\}$ is again trivial. The other cases all follow from the observation that, given multilinear polynomials $p,q$ over variables $x_1,...,x_k$, variable $x$ and $a\in\{\True,\False\}$, all of $1-p$, $\pev{x:=a}p$, $\delta_{x_1}...\delta_{x_k}(p\cdot q)$, and $\delta_{x_1}...\delta_{x_k}(p+q-p\cdot q)$ are multilinear.

\proofparag{Part (b)} If there is a descendant $\psi'$ of $\varphi$, s.t.\ $\Convert(\psi')=\psi$, then the statement follows from part (a), as $\Eval{\psi}$ is multilinear. This leaves two cases. First, if $\psi=\psi_1\circledast\psi_2$, then $\psi_i=\Convert(\psi_i')$ for $i\in\{1,2\}$ by construction. Therefore, $\Eval{\psi_i}$ is multilinear. So if $\circledast=\wedge$ we get $\Eval{\psi}=\Eval{\psi_1}\cdot\Eval{\psi_2}$, which has maximum degree $2$. Analogously for $\circledast=\vee$. 

Second, we have the case $\psi=\delta_x\psi_1$. By induction, we find that $\psi_1$ has maximum degree $2$, which cannot be increased by $\delta_x$.
\Qed
\end{proof}

\section{Proof of Lemma~\ref{lem:numberofsat}}

\numberofsat*

\begin{proof}
Let $X=\{x_1,...,x_n\}:=\Free(\varphi)$, and let $S:=\{\sigma\mid\sigma:X\rightarrow\{0,1\}\}$ be the set of binary assignments on $X$. We also set $p:=\Eval{\Convert(\varphi)}$. Then
\[m\stackrel{(1)}{=}\sum_{\sigma\in S}\Pev{\sigma}\Eval{\varphi}
\stackrel{(2)}{=}\sum_{\sigma\in S}\Pev{\sigma}p
\] 
where (1) uses Proposition~\ref{prop:fundeval} and $m<\Abs{\F}$, and (2) uses $\Eval{\varphi}\equiv_bp=\Eval{\Convert(\varphi)}$ (Lemma~\ref{lem:conv}).

We now introduce the notation $\Sigma_xp:=\pev{x:=0}p+\pev{x:=1}p$ for $x\in X$. Using this notation, we rewrite above equation.
\[m=\Sigma_{x_1}\Sigma_{x_2}\cdots\Sigma_{x_n}p\tag{$*$}\]
Crucially, we can now use the fact that $p$ is multilinear (again Lemma~\ref{lem:conv}), to derive 
\[p=\delta_xp=[1-x]\cdot\pev{x:=0}p+[x]\cdot\pev{x:=1}p\]
for any $x\in X$. Setting $x$ to $1/2$ yields
\begin{gather*}
\pev{x:=1/2}p=\pev{x:=1/2}\big([1-x]\cdot\pev{x:=0}p+[x]\cdot\pev{x:=1}p\big)\\[1mm]\quad
=(\pev{x:=0}p+\pev{x:=1}p)/2=\Sigma_xp/2
\end{gather*}
In other words, $\Sigma_xp=2\cdot\pev{x:=1/2}p$. By plugging this into $(*)$ we get $m=2^n\Pev{\sigma}p$, where $\sigma(x):=1/2$ for $x\in X$, as desired.
\Qed
\end{proof}

\section{Proof of Proposition \ref{prop:algorithm}}

The core of the argument in Proposition~\ref{prop:algorithm} uses the Schwartz-Zippel Lemma,
of which we only need a very simple version.

\begin{lemma}[Schwartz-Zippel Lemma]
\label{lem:sz}
Let $p_1, p_2$ be distinct univariate polynomials over $\F$ of degree at most $d \geq 0$. Let $r$ be selected uniformly at random from $\F$. The probability that $p_1(r) = p_2(r)$ holds is at most $d /\Abs{\F}$.
\end{lemma}
\begin{proof}
Since $p_1\ne p_2$ the polynomial $p := p_1 - p_2$ is not the zero polynomial and has degree at most $d$. Therefore $p$ has at most $d$ zeros, and so the probability of $p(r)=0$ is at most $d / \Abs{\F}$.
\end{proof}

Now we move on to the proof.

\cpcertifycorrect*

\begin{proof}
If the claim is true, 
Prover can always answer every challenge posed by Verifier truthfully. True claims pass all the checks conducted by Verifier, and so Verifier accepts. 

Assume now the claim is false.  We show that Verifier accepts with probability at most $4n\Abs{\varphi}/\Abs{\F}$.

First, we consider the contents of $\Claimset $, the set of claims yet to be checked, after each step of the protocol. This gives rise to a sequence $\Claimset _0, \Claimset _1, ...\Claimset _l$, where $\Claimset _0$ contains only the initial claim. In particular, we consider each iteration of step (b.1) separately, so executing the loop $s$ times adds $s$ elements to the sequence.

As (by assumption) the initial claim is false, $\Claimset_0$ contains a false claim. For the moment, assume that Verifier accepts. They therefore must complete all rounds without rejecting. As the nodes are processed in topological order (every node is processed before its descendants), eventually $\Claimset$ must become empty: an inner node $\psi$ replaces all claims about itself with claims about its children, and each leaf node $\psi$ removes all claims about itself.

So $\Claimset_l=\emptyset$ contains only true claims and thus there must be an $i$, s.t.\ $\Claimset_i$ contains one false claim, but $\Claimset_{i+1}$ contains only true claims. For any such $i$, we say that Prover \emph{tricks} Verifier in step $i$. In other words: if Verifier accepts, it was tricked at some point.

We will now show that at each step, Verifier is tricked either with probability $0$ or probability at most $2/\Abs{\F}$, 
and the latter case occurs at most $2n\Abs{\varphi}$ times. By union bound, this implies the stated bound.

First, we show that Verifier can only be tricked in an iteration of step (b.1), or in step (c.4). Steps (a) and (b.2) do not modify $\Claimset$. The arguments for steps (c.1-c.3) are analogous, so we only present (c.1) with $\circledast=\wedge$ here. Let $\Pev{\sigma}\Eval{\psi} = k$ be the claim to be checked. We have
\[\Pev{\sigma}\Eval{\psi}=\Pev{\sigma}\Eval{\psi_1\wedge\psi_2}
=\Pev{\sigma}(\Eval{\psi_1}\cdot\Eval{\psi_2})
=\Pev{\sigma}\Eval{\psi_1}\cdot\Pev{\sigma}\Eval{\psi_2}\]
So $\Pev{\sigma}\Eval{\psi} = k$ is equivalent to $\Pev{\sigma}\Eval{\psi_1}=k_1\wedge\Pev{\sigma}\Eval{\psi_2}=k_2\wedge k_1k_2=k$. In other words, if $\Pev{\sigma}\Eval{\psi} \ne k$, then either $k_1k_2\ne k$ and Verifier rejects immediately (note $\pev{x \Gets  k_1}\pev{y \Gets k_2}\Eval{x\wedge y}=k_1k_2$), or one of the two claims added to $\Claimset$ is false.

For step (c.4) we get
\[\Pev{\sigma}\Eval{\psi}
=\pev{x:=\sigma(x)}\Pev{\sigma'}\Eval{\delta_x\psi'}
=\pev{x:=\sigma(x)}\Pev{\sigma'}\delta_x\Eval{\psi'}
=\pev{x:=\sigma(x)}\delta_x\Pev{\sigma'}\Eval{\psi'}
\]
so $\Pev{\sigma}\Eval{\psi} = k$ is equivalent to $\Pev{\sigma'}\Eval{\psi'} = p\wedge \pev{x:=\sigma(x)}\delta_xp=k$. Conversely, if the claim does not hold we get either $\pev{x:=\sigma(x)}\delta_xp\ne k$ and Verifier rejects immediately, or $\Pev{\sigma'}\Eval{\psi'} \ne p$. Note that $p$ is a univariate polynomial with degree at most 2 (Lemma~\ref{lem:conv}b), not a constant — Verifier cannot simply add the claim $\Pev{\sigma'}\Eval{\psi'} = p$ to $\Claimset$. However, by Lemma~\ref{lem:sz}, $\Pev{\sigma'}\Eval{\psi'} \ne p$ implies that $\pev{x:=r}\Pev{\sigma'}\Eval{\psi'}=\pev{x:=r}p$ holds with probability at most $2/\Abs{\F}$. Otherwise, the claim added to $\Claimset$ is false and Verifier is not tricked.

Step (b.1) remains. We remark that, though unintuitive, the probability that Verifier is tricked does not increase with $m$, the number of claims to be merged. If all claims $\{ \Pev{\sigma_i}\Eval{\psi} = k_i\}_{i=1}^m$ are true, then clearly Verifier cannot be tricked in this step. So fix an $i$ s.t.\ $\Pev{\sigma_i}\Eval{\psi} \ne k_i$. Similar to step (c.4) above we get that either
$\pev{x \Gets  \sigma_i(x)} \, p_i\ne k_i$ and Verifier rejects, or $\Pev{\sigma_i'}\Eval{\psi} \ne p_i$. In the latter case, $\pev{x\Gets r}\Pev{\sigma_i'}\Eval{\psi} = \pev{x\Gets r}p_i$ again holds with probability at most $2/\Abs{\F}$.

To conclude the proof, we argue that steps (b.1) and (c.4) occur at most $2n\Abs{\varphi}$ times in total. Step (c.4) occurs at most $\Abs{\Convert(\varphi)}\le n\Abs{\varphi}$ times. For (b.1) we note that it is executed at most $n$ times for each node with more than one parent. The conversion of $\varphi$ to a CPD does not increase the number of such nodes, so it also occurs at most $n\Abs{\varphi}$ times.
\Qed
\end{proof}

\section{Proof of Proposition \ref{prop:basic}}

\basicproperties*

\begin{proof}
\proofparag{Part (a)}
We proceed by induction, and show that a BDD $w$ with $i:=\Level(w)$ is multilinear in $x_1,...,x_i$ and does not depend on $x_{i+1},...,x_n$. The base case $i=0$ is trivial. For the induction step, let $w=\Bdd{x_i,u,v}$. We have $\Eval{w}=[1-x_i]\cdot\Eval{u}+[x_i]\cdot\Eval{v}$. By induction hypothesis, $\Eval{u},\Eval{v}$ are multilinear in $x_1,...,x_{i-1}$ and do not depend on $x_i,...,x_n$. The claim follow immediately.

It remains to argue that $\Eval{w}$ is binary, for a BDD $w$. We again proceed by induction on $\Level(w)$. For $w\in\{\Bdd{\False},\Bdd{\True}\}$, this is clear, so let $w=\Bdd{x,u,v}$ and let $\sigma$ denote a binary assignment. We get 
\begin{align*}
\Pev{\sigma}\Eval{w}
&=\Pev{\sigma}\big([1-x_i]\cdot\Eval{u}+[x_i]\cdot\Eval{v}\big)\\
&=[1-\sigma(x_i)]\cdot\Pev{\sigma}\Eval{u}+[\sigma(x_i)]\cdot\Pev{\sigma}\Eval{v}\in\{\Pev{\sigma}\Eval{u},\Pev{\sigma}\Eval{v}\}
\end{align*}
The last step uses $\sigma(x_i)\in\{0,1\}$. By induction hypothesis, both $\Pev{\sigma}\Eval{u}$ and $\Pev{\sigma}\Eval{v}$ are $\Zero$ or $\One$, and the claim follows.

\proofparag{Part (b)}
Before we prove this part, we remark that the statement follows from two well-known facts: multilinear polynomials are uniquely determined by the values they take on binary assignments, and BDDs uniquely represent arbitrary boolean functions. Here, however, we give an elementary proof.

We show that such a $w$ exists if $p$ is a polynomial over variables $x_1,...,x_i$, for all $0\le i\le n$. We proceed by induction on $i$. For $i=0$ we have $p\in\{\Zero,\One\}$ and choose the appropriate $w\in\{\Bdd{\False},\Bdd{\True}\}$. For $i>0$, let $p_b:=\pev{x_i:=b}p$, for $b\in\{0,1\}$, and let $w_b$ denote a BDD for $p_b$.

We have $p=[1-x_i]\cdot p_0+[x_i]\cdot p_1$. If $p_0=p_1$, then $p=p_0$ does not depend on $x_i$, so $p=\Eval{w_0}$. Otherwise, with $w:=\Bdd{x_i,w_0,w_1}$ we have $\Eval{w}=[1-x_i]\cdot\Eval{w_0}+[x_i]\cdot\Eval{w_1}=p$.

It remains to show that $w$ is unique. We will do this by proving that $\Eval{u}\ne\Eval{v}$ for all BDDs $u\ne v$. Assume the contrary, and choose a counterexample $u\ne v$ s.t.\ $\Eval{u}=\Eval{v}$ and $\Level(u)+\Level(v)$ is minimal. Wlog.\ we assume $\Level(u)\le\Level(v)=:i$.

First, we consider the case that $\Level(u)<i$. Then $v=\Bdd{x_i,v_0,v_1}$ and we get $\Eval{u}=\pev{x_i:=b}\Eval{u}=\pev{x_i:=b}\Eval{v}=\Eval{v_b}$ for $b\in\{0,1\}$. Due to $v_0\ne v_1$ (else $v$ would not be a BDD) we find a $b$ with $u\ne v_b$, but $\Eval{u}=\Eval{v_b}$. This contradicts minimality of the counterexample.

Second, we have the case $\Level(u)=\Level(v)=i$. Clearly, $i>0$, so $v=\Bdd{x_i,v_0,v_1}$ and $u=\Bdd{x_i,u_0,u_1}$. Due to $v\ne u$ there is a $b\in\{0,1\}$ with $v_b\ne u_b$. But we get $\Eval{v_b}=\pev{x_i:=b}\Eval{v}=\pev{x_i:=b}\Eval{u}=\Eval{u_b}$, so again a smaller counterexample exists.
\Qed
\end{proof}

\section{Proof of Proposition~\ref{prop:computebdd}}

\computebdd*

The proof will take up the remainder of this section.

For the time bound, observe that \EvaluateEBDD\ performs the same operations as $\bddop$, but with a breadth-first traversal of the BDD, instead of a depth-first one. Clearly, this does not increase the time complexity.

The bound $T\in\O(\Abs{u_1}\cdot\Abs{u_2})$ is a well-known bound for $\bddop$, it relies on there being at most $\Abs{u_1}\cdot\Abs{u_2}$ recursive calls, as each call corresponds to a pair of BDD nodes (and identical calls are memoised). Naturally, the same bound holds for \EvaluateEBDD, where the number of created product nodes is bounded by $\Abs{u_1}\cdot\Abs{u_2}$. Each of these product nodes is operated on once, by replacing it with a BDD nodes.

Now we move to showing correctness, which will follow from Lemmata~\ref{lem:claim3} and~\ref{lem:claim4}. We start with a basic property of the degree reduction operator.

\begin{lemma}\label{lem:degreeredution}
Let $p$ denote a polynomial and $x$ a variable. Then we have $\delta_xp=[1-x]\cdot\pev{x:=0}p+[x]\cdot\pev{x:=1}p$.
\end{lemma}
\begin{proof}
We write $p$ as $p=p_0+[x]\cdot p_1+...+[x^k]\cdot p_k$ for polynomials $p_0,...,p_k$ which do not depend on $x$. We get
\begin{gather*}
[1-x]\cdot\pev{x:=0}p+[x]\cdot\pev{x:=1}p=[1-x]\cdot p_0+[x]\cdot(p_0+p_1+...+p_k)\\
=p_0+[x]\cdot(p_1+...+p_k)=\delta_xp
\end{gather*}
\Qed
\end{proof}

We first show that the innermost loop of \ComputeEBDD\ computes a degree reduction.

\begin{lemma}\label{lem:closed}
Let $\Bdd{u \circledast v}$ be a product eBDD, and let $s:=\Bdd{x_{n-i}, t_0, t_1}$ be the eBDD computed by the innermost loop of \ComputeEBDD\ (Table \ref{table:alg}, left). Then $\Eval{s} = \delta_{x_{n-i}}\Eval{\Bdd{u \circledast v}}$. 
\end{lemma}
\begin{proof}
For $\circledast=\wedge$ and $b\in\{0,1\}$ we get
\begin{align*}
\Eval{t_b}&=\Eval{\Bdd{u_b\wedge v_b}}
=\Eval{u_b}\cdot\Eval{v_b}
=\Eval{\pev{x_{n-i}:=b}u}\cdot\Eval{\pev{x_{n-i}:=b}v}\\
&=\pev{x_{n-i}:=b}\Eval{u}\cdot\pev{x_{n-i}:=b}\Eval{v}
=\pev{x_{n-i}:=b}\big(\Eval{u}\cdot\Eval{v}\big)
=\pev{x_{n-i}:=b}\Eval{\Bdd{u\wedge v}}
\end{align*}
Analogously, one can derive $\Eval{t_b}=\pev{x_{n-i}:=b}\Eval{\Bdd{u\circledast v}}$ for $\circledast=\vee$ as well. Finally:
\begin{align*}
\Eval{s}
&=[1-x_{n-i}]\cdot\Eval{t_0}+[x_{n-i}]\cdot\Eval{t_1}\\
&=[1-x_{n-i}]\cdot\pev{x_{n-i}:=0}\Eval{\Bdd{u\circledast v}}+[x_{n-i}]\cdot\pev{x_{n-i}:=1}\Eval{\Bdd{u\circledast v}}\\
&=\delta_{x_{n-i}}\Eval{\Bdd{u \circledast v}}
\end{align*}
The last step uses Lemma~\ref{lem:degreeredution}.\Qed
\end{proof}

Now we show some simple invariants of the algorithms.

\begin{lemma}\label{lem:claim1}
$w_i$ only has product nodes at levels $1,...,n-i$.
\end{lemma}
\begin{proof}
For $w_0$ the statement holds vacuously. Assume the statement holds for $w_i$. The algorithm computes $w_{i+1}$ by replacing each product node $\Bdd{u \circledast v}$ of $w_i$ at level $n-i$ with a non-product node $\Bdd{x_{n-i}, t_0,t_1}$. So $w_{i+1}$ has no product nodes at level $n-i$.\Qed
\end{proof}

\begin{lemma}\label{lem:claim2}
$\Eval{w_i}$ is multilinear in all of $x_{n-i+1},...,x_n$.
\end{lemma}
\begin{proof}
We prove that for every level $j\in\{1,...,n\}$ and every node $v$ of $w_i$ at level $j$ the polynomial $\Eval{v}$ is a multilinear polynomial over the variables $x_{n-i+1},...,x_n$. We proceed by induction on $j$. The case $j=0$ is trivial. For the inductive step, consider two cases. If $v$ is a product node, then $j\le n-i$ by Lemma~\ref{lem:claim1}. By Definition~\ref{def:eBDD}, $\Eval{v}$ is a polynomial over $x_1,...,x_j$. So $\Eval{v}$ does not depend on $x_{n-i+1},...,x_n$, and in particular $\Eval{v}$ is multilinear in them. If $v$ is not a product node, then $v=\Bdd{x_j,v_0,v_1}$ for nodes $v_0, v_1$ and $\Eval{v}=[1-x_j]\cdot\Eval{v_0}+[x_j]\cdot\Eval{v_1}$ (Definition~\ref{def:BDD}). Again by Definition~\ref{def:BDD} $\Eval{v_0},\Eval{v_1}$ do not depend on $x_j$, and by induction hypothesis, they are multilinear in $x_{n-i+1},...,x_n$. So $\Eval{v}$ is multilinear in $x_{n-i+1},...,x_n$ as well.\Qed
\end{proof}

\begin{lemma}\label{lem:claim3}
$\Eval{w_{i+1}}=\delta_{x_{n-i}}\Eval{w_{i}}$.
\end{lemma}
\begin{proof}
For any node $v$ of $w_i$, let $v^*$ denote the corresponding node in $w_{i+1}$. More precisely, let $u_1,v_1,...,u_l,v_l$ denote the sequence of replaced nodes, i.e.\ $w_{i+1}=w_i[u_1/v_1]\cdots[u_l/v_l]$. Then $v^*:=v[u_1/v_1]\cdots[u_l/v_l]$. So if $v$ is a product node at level $n-i$, $v^*$ denotes the BDD node that replaces it, and for the other nodes, $v^*$ and $v$ are the same node, except that some descendant of $v$ have been replaced by new nodes. Note that $w_i^*=w_{i+1}$, and that $\Bdd{x,v_0,v_1}^*=\Bdd{x,v_0^*,v_1^*}$.

We prove the stronger claim that $\Eval{v^*}=\delta_{x_{n-i}}\Eval{v}$ for every descendant $v$ of $w_{i+1}$. We proceed by induction. For the two leaves $\True$ and $\False$ the statement clearly holds. For the induction step, we consider two cases.
\begin{itemize}
\item $v$ is a product node. If $\Level(v)<n-i$, then $v^*=v$. By Lemma~\ref{lem:claim2}, $\Eval{v}$ is multilinear in $x_{n-i}$, and so $\Eval{v^*}=\delta_{x_{n-i}}\Eval{v}=\Eval{v}$. If $\Level(v)=n-i$, then $\Eval{v^*}=\delta_{x_{n-i}}\Eval{v}$ follows from Lemma~\ref{lem:closed}.
\item $v=\Bdd{x,v_0,v_1}$. As noted above, we have $v^*=\Bdd{x,v_0^*,v_1^*}$, so
\begin{gather*}
\delta_{x_{n-i}}\Eval{v}\stackrel{(1)}{=}\delta_{x_{n-i}}\big([1-x]\cdot\Eval{v_0}+[x]\cdot\Eval{v_1}\big)\\
\quad \stackrel{(2)}{=}[1-x]\cdot\delta_{x_{n-i}}\Eval{v_0}+[x]\cdot\delta_{x_{n-i}}\Eval{v_1}\\
\quad \stackrel{(3)}{=}[1-x]\cdot\Eval{v_0^*}+[x]\cdot\delta_{x_{n-i}}\Eval{v_1^*}
\stackrel{(4)}{=}\Eval{v^*}
\end{gather*}
where (1) expands the arithmetisation, (2) uses that either $x\ne x_{n-i}$, or $\Eval{v_b}$ does not depend on $x_{n-i}$ and $\delta_{x_{n-i}}\Eval{v_0}=\Eval{v_0}$, for $b\in\{0,1\}$, (3) uses the induction hypothesis, and (4) folds the arithmetisation back.
\end{itemize} \Qed
\end{proof}

\begin{lemma}\label{lem:claim4}
$w_n$ is a BDD and $w_n=\bddop(u_1,u_2)$.
\end{lemma}
\begin{proof}
By Lemma~\ref{lem:claim1}, $w_n$ has no product nodes. By Definition~\ref{def:eBDD}, an eBDD without product nodes is a BDD.
Moreover,
\[\Eval{\bddop(u_1,u_2)}\equiv_b\Eval{\psi_1\circledast\psi_2}=\Eval{w_0}\equiv_b\delta_{x_n}\Eval{w_0}=\Eval{w_1}\equiv_b...\equiv_b\Eval{w_n}\]
By Proposition~\ref{prop:basic}a, both $\Eval{\bddop(u_1,u_2)}$ and $\Eval{w_n}$ are multilinear polynomials. It is well-known that two multilinear polynomials that coincide on all binary inputs must be equal, so we get $\Eval{\bddop(u_1,u_2)}=\Eval{w_n}$, and by Proposition~\ref{prop:basic}b, $\bddop(u_1,u_2)=w_n$
~\Qed
\end{proof}

\section{Proof of Proposition~\ref{prop:evaluateebdd}}

\evaluateebdd*
\begin{proof}
The algorithm memoises the computed polynomials for each node, so the total number of calls is in $\O(\Abs{w})$. In each call, a constant number of operations on polynomials are performed. For any eBDD $w$, the polynomial $\Eval{w}$ has maximum degree at most $2$. (This follows immediately from Definition~\ref{def:eBDD}: a product node can only have two BDDs as children, not eBDDs.) A polynomial with $n-k$ free variables and maximum degree at most $2$ can be represented using $3^{n-k}$ coefficients (one for each monomial), and operations can be performed efficiently on this representation. 
\end{proof}

\section{Proof of Theorem~\ref{thm:main}}

\thmmain*
\begin{proof}
\proofparag{Part (a)} This follows immediately from Proposition~\ref{prop:computebdd}.

\newcommand{\Sbdd}{S_{\mathrm{bdd}}}
\proofparag{Part (b)}
Let $S$ denote the set of descendants of $\Convert(\varphi)$, and set
\[\Sbdd:=\{\Convert(\psi):\psi\text{ descendant of }\varphi\}\subseteq S\]
to the descendants that correspond to BDDs (and not eBDDs). Note $\Abs{\Sbdd}=\Abs{\varphi}$ and $\Abs{S}\le n\Abs{\varphi}$. Additionally, let $B_\psi$ denote the eBDD representing $\psi\in S$, computed by \ComputeEBDD. Note that $B_\psi$ is a BDD if $\psi\in\Sbdd$, and (as BDDs are unique, see Proposition~\ref{prop:basic}) those necessarily match the BDDs computed by \BDDSolver. We thus observe $\sum_{\psi\in\Sbdd}\Abs{B_\psi}\le T$.

For the eBDD $\psi\in S\setminus\Sbdd$, each node also appears in the computation of \ComputeEBDD. In the sum $\sum_{\psi\in S}\Abs{B_\psi}$, however, a node is counted up to $n$ times, so we get $\sum_{\psi\in S}\Abs{B_\psi}\le nT$.

As shown in Proposition~\ref{prop:evaluateebdd}, responding to one challenge takes time linear in $\Abs{w}$, where $w$ is the evaluated eBDD. Step (b.1.1) of \CPCertify\ sends at most $n$ challenges for each node in $\Sbdd$, which are evaluated in time linear in $\sum_{\psi\in\Sbdd}nf\Abs{B_\psi}\le nT$. Step (c) sends a challenge for each node in $S$, which take at most $\sum_{\psi\in S}\Abs{B_\psi}\le nT$ time.

\proofparag{Part (c)}
As argued for part (b), Verifier sends at most $n\Abs{\varphi}$ challenges. The challenge consist of one partial assignment, which has size at most $n$. The failure probability follows from Proposition~\ref{prop:algorithm}.
\end{proof}

\section{Evaluation – Detailed Description}\label{app:setup}
\subsection{Instances}
We used the instances from the crafted instances track of the QBF Evaluation 2022 (\url{http://www.qbflib.org/QBFEVAL_20_DATASET.zip}). These are re-used from the QBF Evaluation 2020. It should be noted that these instances are all unsatisfiable. (Instances in QDIMACS format are specified so that the outermost quantifier is existential.)

We also use the linear domino placement game used for evaluating \pgbddq{}~\cite{BryantH21qbf}. They can be obtained at \url{https://github.com/rebryant/pgbddq-artifact}. We reproduce the parameters of~\cite[Table 3]{BryantH21qbf}. On these instances we run only \pgbddq\ and our tool, which are both BDD-based, and we allow both to use the provided variable ordering, which improves performance significantly.

\subsection{Tools}
The following four tools were compared:
\begin{center}
\setlength{\tabcolsep}{1.5mm}
\begin{tabular}{llll}
tool&language&version& link \\ \midrule
\caqe{} & Rust & \texttt{8b646df} & \url{https://github.com/ltentrup/caqe} \\
\depqbf{} & C & \texttt{2ad3995} & \url{https://github.com/lonsing/depqbf} \\
\pgbddq{} & Python & \texttt{d5cbc96}&\url{https://github.com/rebryant/pgbdd} \\
\qchecker{} & Python & \texttt{d5cbc96}&\url{https://github.com/rebryant/pgbdd} \\
\qrpcheck{} & C & 1.0.3&\url{http://fmv.jku.at/qrpcheck/} \\
\blic{} (ours) & C++ & \texttt{bd3d298} & \url{https://gitlab.lrz.de/i7/blic} \\
\end{tabular}
\end{center}

All solvers were run without a preprocessor, limiting their performance. In the QBF Evaluation 2022, \caqe{} combined with the preprocessor \bloqqer{} achieved first place in the crafted instances track. For comparison, we ran our tool and \caqe{} on the preprocessed instances: \bloqqer{} solves 97 of 172 by itself, of the remaining 74 instances our tool solves 34, while \caqe{} solves 50 (timeout of 10min). 

We ran \depqbf{} with certificate generation enabled. More precisely, we used flags
\begin{center}
\begin{minipage}{0.9\textwidth}
\texttt{--trace=bqrp} \texttt{--dep-man=simple} \texttt{--no-lazy-qpup} \texttt{--no-dynamic-nenofex} \texttt{--no-qbce-dynamic} \texttt{--no-trivial-falsity} \texttt{--no-trivial-truth}
\end{minipage}
\end{center}
These flags disable features that are not supported in conjunction with certificate generation. This reduces the performance of \depqbf{}: in its default configuration it can solve 104 of 172 instances (compared to 87 when certification was enabled).

To verify certificates generated by \pgbddq{} we use the tool \qchecker{}, which is part of \pgbddq{}. It is specialised for the certificates \pgbddq{} generates.

\subsection{Time Measurement}
Times are measured in the calling process. These times exceed self-reported times by about 1-2ms. Running our tool repeatedly on one (arbitrarily chosen) instance yields an average run-to-run deviation of 11ms, maximum deviation of 34ms (compared to a total running time of 3.08s).

To measure the time taken by our tool for Verifier and Prover parts of \CPCertify, it is necessary to measure the contributions of each round of the interactive protocol. As the protocol is executed within a single process, these data are collected internally in our tool.

\subsubsection{Comparison with results in~\cite{BryantH21qbf}.}
As we use the same instances and configuration of \pgbddq{}, we can compare the times we obtained with the times in~\cite[Table 3]{BryantH21qbf} to verify that we can reproduce their numbers.

Our results match their results relatively closely. Our times are between $15\%$ and $29\%$ slower for $10\le N\le 25$, and $45-49\%$ slower for $N=45$, with one outlier at $58\%$ (a subsequent run was $54\%$ slower, making it unlikely that the issue is intermittent). This can likely be accounted for by the faster (in terms of single-thread performance) Intel Core i7-7700K processor used in~\cite{BryantH21qbf}, and differences in main memory.

\end{document}